\documentclass[10pt,sigconf]{acmart}

\renewcommand\footnotetextcopyrightpermission[1]{} 
\def\fullversion
\ifdefined\fullversion
\else
\def\shortversion
\fi

\usepackage{xcolor}
\usepackage{multirow}
\usepackage{hhline}
\usepackage[normalem]{ulem}
\usepackage{url}
\makeatletter
\g@addto@macro{\UrlBreaks}{\UrlOrds}
\makeatother

\usepackage{algorithm}
\usepackage[noend]{algpseudocode}
\usepackage[caption=false]{subfig}
\usepackage{enumitem} 
\usepackage{xspace}
\usepackage{latexsym}
\usepackage{amsmath}
\usepackage[usestackEOL]{stackengine}
\usepackage{amssymb}
\usepackage{amsthm}
\usepackage{tabularx}
\usepackage{siunitx}
\usepackage{hyperref}
\usepackage{cleveref}

\crefformat{section}{\S#2#1#3}
\Crefformat{section}{\S#2#1#3}

\usepackage{booktabs}
\usepackage{colortbl}
\usepackage{float}
\usepackage{listings}
\usepackage{mathtools}
\usepackage{caption}

\usepackage{xcolor}
\definecolor{shadecolor}{RGB}{222,222,200}
\usepackage{tikz}
\usepackage{listings}

\makeatletter
\newenvironment{btHighlight}[1][]
{\begingroup\tikzset{bt@Highlight@par/.style={#1}}\begin{lrbox}{\@tempboxa}}
{\end{lrbox}\bt@HL@box[bt@Highlight@par]{\@tempboxa}\endgroup}

\newcommand\btHL[1][]{%
  \begin{btHighlight}[#1]\bgroup\aftergroup\bt@HL@endenv%
}
\def\bt@HL@endenv{%
  \end{btHighlight}%
  \egroup
}
\newcommand{\bt@HL@box}[2][]{%
  \tikz[#1]{%
    \pgfpathrectangle{\pgfpoint{1pt}{0pt}}{\pgfpoint{\wd #2}{\ht #2}}%
    \pgfusepath{use as bounding box}%
    \node[anchor=base west ,outer sep=0pt,inner xsep=1pt, inner ysep=0pt, rounded corners=3pt, minimum height=\ht\strutbox+1pt,#1]{\raisebox{1pt}{\strut}\strut\usebox{#2}};
  }%
}
\makeatother

\lstdefinestyle{SQL}{
    language={SQL},basicstyle=\ttfamily, 
    moredelim=**[is][\btHL]{`}{`},
    moredelim=**[is][{\btHL[]}]{@}{@},
}
\usepackage{soul}
\usepackage{xcolor}

\makeatletter
\def\@highlightttpeeknext{\futurelet\@nexttoken\@highlightttaux}
\def\@highlighttt #1.{%
    \def\@highlightttaux{\ifx\@nexttoken\egroup
       \myhighlightmethod {#1}\else
       \myhighlightmethod {#1.}\linebreak[2]%
       \expandafter\@highlighttt\fi}%
    \@highlightttpeeknext}

\def\@plaintt {\futurelet\@nexttoken\@plainttaux}
\def\@plainttaux {\ifx\@nexttoken\egroup\else
                  \ifx\@nexttoken\bgroup
                  \expandafter\expandafter\expandafter\@plaintta\else
                  \expandafter\expandafter\expandafter\@plainttb\fi\fi}
\def\@plaintta #1{{#1}\@plaintt}
\def\@plainttb #1{\ifcat\@nexttoken a\penalty\hyphenpenalty \plaintthook
  #1\else \plaintthook{#1}\linebreak[2]\fi\@plaintt}

\setlist{nolistsep}
\DeclareSIUnit[scientific-notation=engineering,prefixes-as-symbols=false]{\ns}{\nano\second}

\newcommand{\NAME}{{Cheetah}\xspace}

\usepackage{titlesec}
\titlespacing*{\section}
{0pt}{.6ex plus .2ex minus .1ex}{.5ex plus .1ex}
\titlespacing*{\subsection}
{0pt}{.45ex plus .1ex minus .1ex}{0.3ex plus .1ex}
\titlespacing*{\subsubsection}
{0pt}{.2ex plus .1ex minus .1ex}{0.2ex plus .1ex}
\titlespacing*{\paragraph}
{0pt}{.1ex plus .1ex minus .1ex}{0.1ex plus .1ex}

\newcommand{\parab}[1]{\noindent\textbf{#1}}
\newcommand{\parabe}[1]{\noindent\textbf{\em #1}}
\newcommand{\cut}[1]{}

\newcommand{\sigmodrev}[1]{{#1}}

\newcommand{\turnOffComments}{}
\ifdefined\turnOffComments
\newcommand{\ran}[1]{}
\newcommand{\minlan}[1]{}
\newcommand{\tirmazi}[1]{}
\else
\newcommand{\ran}[1]{\textcolor{purple}{(Ran: #1)}}
\newcommand{\minlan}[1]{\textcolor{blue}{(minlan: #1)}}
\newcommand{\tirmazi}[1]{\textcolor{violet}{(tirmazi: #1)}}
\fi

\newcommand{\skyline}{{\sc SKYLINE}\xspace}

\newcommand{\topn}{{\sc TOP N}\xspace}
\newcommand{\groupby}{{\sc GROUP BY}\xspace}
\newcommand{\join}{{\sc JOIN}\xspace}
\newcommand{\having}{{\sc HAVING}\xspace}
\newcommand{\distinct}{{\sc DISTINCT}\xspace}

\newcounter{algsubstate}

\newcommand{\name}{\ensuremath{\mathit{name}}\xspace}
\newcommand{\seller}{\ensuremath{\mathit{seller}}\xspace}
\newcommand{\price}{\ensuremath{\mathit{price}}\xspace}
\newcommand{\taste}{\ensuremath{\mathit{taste}}\xspace}
\newcommand{\texture}{\ensuremath{\mathit{texture}}\xspace}
\newcommand{\products}{\ensuremath{\mathit{Products}}\xspace}
\newcommand{\ratings}{\ensuremath{\mathit{Ratings}}\xspace}

\newcommand{\inputRowTerm}{entry}
\newcommand{\inputRowTerms}{entries}
\newcommand{\inputColTerm}{column}
\newcommand{\inputColTerms}{columns}
\newcommand{\matrixRowTerm}{row}
\newcommand{\matrixRowTerms}{rows}

\newcommand{\matrixColTerms}{columns}
\newcommand{\inputRow}{\inputRowTerm\xspace}

\newcommand{\inputRows}{\inputRowTerms{}\xspace}
\newcommand{\outputRows}{output \inputRowTerms{}\xspace}
\newcommand{\inputCol}{input \inputColTerm\xspace}
\newcommand{\inputCols}{input \inputColTerms\xspace}
\newcommand{\matrixRow}{\matrixRowTerm\xspace}
\newcommand{\matrixRows}{\matrixRowTerms\xspace}

\newcommand{\matrixCols}{matrix \matrixColTerms\xspace}
\newcounter{insightlabel}
\newcounter{insightnmbr}
\renewcommand{\theinsightlabel}{\textbf{\theinsightnmbr}}

\newtheorem{theorem}{Theorem}
\newtheorem{fact}{Fact}

\newtheorem{lemma}{Lemma}

\copyrightyear{2020}
\acmYear{2020}
\setcopyright{acmlicensed}
\acmConference[SIGMOD'20]{Proceedings of the 2020 ACM SIGMOD International Conference on Management of Data}{June 14--19, 2020}{Portland, OR, USA}
\acmBooktitle{Proceedings of the 2020 ACM SIGMOD International Conference on Management of Data (SIGMOD'20), June 14--19, 2020, Portland, OR, USA}
\acmPrice{15.00}
\acmDOI{10.1145/3318464.3389698}
\acmISBN{978-1-4503-6735-6/20/06}

\begin{document}

\date{}
\author{Muhammad Tirmazi, Ran Ben Basat, Jiaqi Gao, Minlan Yu}
\affiliation{Harvard University}
\email{{tirmazi,ran,jiaqigao,minlanyu}@g.harvard.edu}

\title{Cheetah: Accelerating Database Queries \\ with Switch Pruning
} 
\makeatletter
\patchcmd{\@maketitle}
  {\addvspace{0.5\baselineskip}\egroup}
  {\addvspace{-2\baselineskip}\egroup}
  {}
  {}
\makeatother

\newcommand{\red}[1]{\textcolor{red}{#1}}

\newcommand{\lastIDX}{\ensuremath{k\cdot n + m}}
\newcommand{\newInputLetter}{\ensuremath{{\rho}}}
\newcommand{\newInput}{\ensuremath{\bar{\newInputLetter}}}
\newcommand{\UBnItemsPerBlock}{\max\set{\floor{\mu},1}}
\newcommand{\UBnBlocks}{\floor{n/\nu}}

\newcommand{\bSize}{\ceil{\mu}}
\newcommand{\bContentSize}{\max\set{\floor{\mu^{-1}},1}}
\newcommand{\nBlocks}{\floor{n/\bSize}}
\newcommand{\lnrLBVal}{\nBlocks \log\big({\bContentSize + 1}\big)}
\newcommand{\lnrLB}{\lnrLBVal = \floor{\frac{n}{\ceil{\lnrErrSymbol / \ell}}} \log\big({\max\set{\floor{\ell / \lnrErrSymbol},1} + 1}\big)}
\newcommand{\lnrLBSymbol}{\mathcal B_{\ell,n,\lnrErrSymbol}}
\newcommand{\attention}{\$\$\$}
\newcommand{\set}[1]{\left\{#1\right\}}
\newcommand{\angles}[1]{\left\langle#1\right\rangle}
\newcommand{\ecor}{$\epsilon$-correct}
\newcommand{\minDelta}{1}
\newcommand{\qsr}[1][\minDelta]{\ensuremath{(\ell,n,#1)}-\emph{Sliding Ranker}}
\newcommand{\lnrErrSymbol}{\Delta}
\newcommand{\sensitivity}{\ensuremath{\widetilde{\lnrErrSymbol}}}
\newcommand{\lnr}[1][\lnrErrSymbol]{\ensuremath{(\ell,n,#1)}-\emph{Ranker}}
\newcommand{\approxSet}{\ensuremath{(1,n,\lnrErrSymbol)}-\emph{Ranker}}
\newcommand{\ssr}{\emph{$\ceil{\frac{W}{S}}$-Sliding Ranker}}
\newcommand{\last}{\ensuremath{\mathit{last}}}
\newcommand{\total}{\ensuremath{\mathit{total}}}
\newcommand{\subTotal}{\ensuremath{\mathit{subTotal}}}
\newcommand{\poly}{\mbox{poly}}
\newcommand{\ceil}[1]{ \left\lceil{#1}\right\rceil}
\newcommand{\floor}[1]{ \left\lfloor{#1}\right\rfloor}
\newcommand{\bfloor}[1]{ \big\lfloor{#1}\big\rfloor}
\newcommand{\Bfloor}[1]{ \Big\lfloor{#1}\Big\rfloor}
\newcommand{\parentheses}[1]{ \left({#1}\right)}
\newcommand{\lnp}[1]{ \ln\parentheses{#1}}
\newcommand{\bParentheses}[1]{ \big({#1}\big)}
\newcommand{\BParentheses}[1]{ \Big({#1}\Big)}
\newcommand{\biggParentheses}[1]{ \bigg({#1}\bigg)}
\newcommand{\abs}[1]{ \left|{#1}\right|}
\newcommand{\logp}[1]{\log\parentheses{#1}}
\newcommand{\Omegap}[1]{\Omega\parentheses{#1}}
\newcommand{\Op}[1]{O\parentheses{#1}}
\newcommand{\Thetap}[1]{\Theta\parentheses{#1}}
\newcommand{\omegap}[1]{\omega\parentheses{#1}}
\newcommand{\logc}[1]{\log\ceil{#1}}
\newcommand{\logf}[1]{\log\floor{#1}}
\newcommand{\flogp}[1]{\floor{\logp{#1}}}
\newcommand{\clogp}[1]{\ceil{\logp{#1}}}
\newcommand{\clogc}[1]{\ceil{\logc{#1}}}
\newcommand{\cdotpa}[1]{\cdot\parentheses{#1}}
\newcommand{\inc}[1]{$#1 = #1 + 1$}
\newcommand{\oneOverE}{ \eps^{-1} }
\newcommand{\oneOverT}{ \tau^{-1} }
\newcommand{\range}[2][0]{#1,1,\ldots,#2}
\newcommand{\orange}[1]{\set{1,2,\ldots,#1}}
\newcommand{\frange}[1]{\set{\range{#1}}}
\newcommand{\xrange}[1]{\frange{#1-1}}
\newcommand{\smallMultError}{(1+o(1))}
\newcommand{\sFactor}{(1+o(1))}
\newcommand{\sNegFactor}{(1-o(1))}
\newcommand{\brackets}[1]{\left[#1\right]}
\newcommand{\lowerbound}{\max \set{\log W ,\frac{1}{2\epsilon+W^{-1}}}}
\newcommand{\smallEpsLowerbound}{\window\logp{\frac{1}{\weps}}}
\newcommand{\smallEpsMemoryTheta}{$\Theta\parentheses{\smallEpsMemoryConsumption}$}
\newcommand{\smallEpsMemoryConsumption}{W\cdot\logp{\frac{1}{\weps}}}

\newcommand{\nBits}{\mathfrak b}
\newcommand{\ranker}{\ensuremath{\mathfrak R}}
\newcommand{\suff}{$(W,\eps)$-Suffix Summer}
\newcommand{\numBits}{\ensuremath{\mathit{numElems}}}
\newcommand{\setBits}{\ensuremath{\mathit{totalSum}}}
\newcommand{\lastBit}{\ensuremath{\mathit{oldest_\newInputLetter}}}
\newcommand{\outBits}{\ensuremath{\mathit{out}}}
\newcommand{\outBitsVal}{\ensuremath{\parentheses{\nu - \parentheses{(i-o)\mod \nu}}}}

\newcommand{\largeEpsRestriction}{For any \largeEps{},}
\newcommand{\largeEps}{\ensuremath{\eps^{-1} \le 2W\left(1-\frac{1}{\logw}\right)}}
\newcommand{\smallEpsRestriction}{For any \smallEps{},}
\newcommand{\smallEps}{$\eps^{-1}>2W\left(1-\frac{1}{\logw}\right)=2\window(1-o(1))$}
\newcommand{\bc}{{\sc Basic-Counting}}
\newcommand{\bs}{{\sc Basic-Summing}}
\newcommand{\windowcounting}{ {\sc $(W,\epsilon)$-Window-Counting}}

\newcommand{\query}[1][] { {\sc Query}$(#1)$}
\newcommand{\add}  [1][] { {\sc Add}$(#1)$}

\newcommand{\window}{n}
\newcommand{\logw}{\log \window}
\newcommand{\logrw}[1][]{\log^{#1}{\bsrange\window}}
\newcommand{\flogw}{\floor{\log \window}}
\newcommand{\weps}{\window\epsilon}
\newcommand{\logweps}{\log{\weps}}
\newcommand{\bitarray}{b}
\newcommand{\currentBlockIndex}{i}
\newcommand{\currentBlock}{\bitarray_{\currentBlockIndex}}
\newcommand{\remainder}{y}
\newcommand{\numBlocks}{k}
\newcommand{\sumOfBits}{B}
\newcommand{\blockSize}{S}
\newcommand{\iblockSize}{\frac{\numBlocks}{\window}}
\newcommand{\threshold}{\blockSize}
\newcommand{\halfBlock}{\frac{\window}{2\numBlocks}}
\newcommand{\blockOffset}{o}
\newcommand{\inputVariable}{x}

\newcommand{\bcTableColumnWidth}{1.5cm}
\newcommand{\bsTableColumnWidth}{1.7cm}
\newcommand{\bsExtendedTableColumnWidth}{3cm}
\newcommand{\bcExtendedTableColumnWidth}{2.8cm}
\newcommand{\bcNarrowTableColumnWidth}{1.5cm}
\newcommand{\bsNarrowTableColumnWidth}{1.5cm}
\newcommand{\bsWorstCaseTableColumnWidth}{2cm}

\newcommand{\bsrange}{ \ell }
\newcommand{\bsReminderPercisionParameter}{ \gamma }
\newcommand{\bsest}{ \widehat{\bssum}}
\newcommand{\bssum}{ S^W }
\newcommand{\bsFracInput}{ \inputVariable' }
\newcommand{\bserror}{ \bsrange\window\epsilon }
\newcommand{\bsfractionbits}{ \frac{\bsReminderPercisionParameter}{\epsilon} }
\newcommand{\bsReminderFractionBits}{ \upsilon}
\newcommand{\bsAnalysisTarget}{ \bssum}
\newcommand{\bsAnalysisEstimator}{ \widehat{\bsAnalysisTarget}}
\newcommand{\bsAnalysisError}{ \bsAnalysisEstimator - \bsAnalysisTarget}
\newcommand{\bsRoundingError}{ \xi}

\setlength{\parskip}{0em}
\newcommand{\paraspace}{\vspace{0.1em}}

\fancyhead[]{}

\vspace{-12mm}

\begin{abstract}
Modern database systems are growing increasingly distributed and struggle to reduce query completion time with a large volume of data. 
In this paper, we leverage programmable switches in the network to partially offload query computation to the switch. While switches provide high performance, they have resource and programming constraints that make implementing diverse queries difficult. To fit in these constraints, we introduce the concept of data \emph{pruning} -- filtering out entries that are guaranteed not to affect output. The database system then runs the same query but on the pruned data, which significantly reduces processing time. We propose pruning algorithms for a variety of queries. 
We implement our system, Cheetah, on a Barefoot Tofino switch and Spark. Our evaluation on multiple workloads shows $40 - 200\%$ improvement in the query completion time compared to Spark.

\end{abstract}
\maketitle
\section{Introduction}

Database systems serve as the foundation for many applications such as data warehousing, data analytics, and business intelligence~\cite{thusoo2010data}. 
Facebook is reported to run more than 30,000 database queries that scan over a petabyte per day ~\cite{prestointeractfacebook}.
With the increase of workloads, the challenge for database systems today is providing high performance for queries on a large distributed set of data.

A popular database query processing system is Spark SQL~\cite{sparksql}. Spark SQL optimizes query completion time by assigning \emph{tasks} to workers (each working on one data partition) and aggregating the query result at a master worker. Spark maximizes the parallelism and minimizes the data movement for query processing. For example, each worker sends a stream of the resulting metadata (e.g., just the columns relevant to the query) before sending the entire rows that are requested by the master. Despite the optimizations, the query performance is still limited by software speed.

We propose \NAME, a query processing system that partially offloads queries to programmable switches. Programmable switches are supported by major switch vendors~\cite{tofino,tofinov2,Trident,XPliant}. They allow programmatic processing of multiple Tbps of traffic~\cite{tofino}, which is orders of magnitude higher throughput than software servers and alternative hardware such as FPGAs and GPUs. Moreover, switches already sit between the workers and thus can process aggregated \mbox{data across partitions.}

However, we cannot simply offload all database operations to switches as they have a constrained programming model~\cite{RMT}: Switches process incoming packets in a pipeline of stages. At each stage, there is a limited amount of memory and computation. Further, there is a limited number of bits we can transfer across stages. 
These constraints are at odds with the large amount of data, diverse query functions, and many intermediate states in database systems.

To meet these constraints, we propose a new abstraction called \emph{pruning}. Instead of offloading full functionality to programmable switches, we use the switch to prune a large portion of data based on the query, and the master only needs to process the remaining data in the same way as it does without the switch.
For example, the switch may remove some duplicates in a \distinct query, and let the master remove the rest, thus accelerating the query.

The pruning abstraction allows us to design algorithms that can fit with the constrained programming model at switches: First, we do not need to implement all database functions on the switches but only offload those that fit in the switch's programming model. Second, to fit in the limited memory of switches, we either store a cached set of results or summarized intermediate results while ensuring a high pruning rate. Third, to reduce the number of comparisons, we use in-switch partitioning of the data such that each \inputRow{} is only compared with a small number of \inputRows{} in its partition. We also use projection techniques that map high-dimensional data points into scalars which allows an efficient comparison.

Based on the pruning abstraction, we design and develop multiple query algorithms ranging from filtering and \distinct to more complex operations such as \join or \groupby. Our solutions are rigorously analyzed and we prove bounds on the resulting pruning rates. We build a prototype on the Barefoot Tofino programmable switch~\cite{tofino} and demonstrate $40 - 200\%$ reduction of query completion times compared with Spark SQL. 



\section{Using programmable switches}\label{sec:sys_overview}

\NAME leverages programmable switches to reduce the amount of data transferred to the query server and thus improves query performance. Programmable switches that follow the PISA model consist of multiple pipelines through which a network packet passes sequentially. These pipelines contain stages with disjoint memory which can do a limited set of operations as the packet passes through them. See~\cite{PRECISION} for more information about the limitations. In this section, we discuss the benefits and constraints of programmable switches to motivate our {\em pruning abstraction}.

\subsection{Benefits of programmable switches}
We use Spark SQL as an example of a database query execution engine. Spark SQL is widely used in industry~\cite{sparksql}, adopts common optimizations such as columnar memory-optimized storage and vectorized processing, and has comparable performance to Amazon's RedShift ~\cite{amplabbenchmarkruns} \and Google
s BigQuery~\cite{bigqueryredshiftbenchmark}.

When a user submits a query, Spark SQL uses a Catalyst optimizer that generates the query plan and operation code (in the form of tasks) to run on a cluster of workers. The worker fetches a partition from data sources, runs its assigned task, and passes the result on to the next batch of workers. At the final stage, the \emph{master} worker aggregates the results and returns them to the user. Spark SQL optimizes query completion time by having workers process data in their partition as much as possible and thus minimizes the data transferred to the next set of workers. As a result, the major portion of query completion time is spent at the tasks the workers run. Thus, Spark SQL query performance is often bottlenecked by the server processing speed and not the network~\cite{sparksql}.

\NAME rides on the trend of significant growth of the network capacity (e.g., up to 100Gbps or 400Gbps per port) and the advent of programmable switches, which are now provided by major switch vendors (e.g., Barefoot~\cite{tofino,tofinov2}, Broadcom~\cite{Trident}, and  Cavium Xpliant~\cite{XPliant}). These switches can process billions of packets per second, already exist in the data path, and thus introduce no latency or additional cost. Table~\ref{tab:tofino} compares the throughput and delay of Spark SQL on commodity servers with those of programmable Tofino switches. The best throughput of servers is 10-100Gbps, but switches can reach 6.5-12.8 Tbps. The switch throughput is also orders of magnitudes better than alternative hardware choices such as FPGAs, GPUs, and smart NICs. 
Switches also have less than 1 $\mu$s delay per packet.

These switches already exist in the cluster and already see the data transferred between the workers. They are at a natural place to help process queries. By offloading part of the processing to switches, we can reduce the workload at workers and thus significantly reduce the query completion time despite more data transferred in the network. Compared with server-side or storage-side acceleration, switches have the extra benefit in that it can process the aggregated data across workers. We defer the detailed comparison of \NAME with alternative hardware solutions to~\cref{sec:related}.

\subsection{Constraints of programmable switches}
Programmable switches make it possible to offload part of queries because they parse custom packet formats and thus can understand data block formats. Switches can also store a limited amount of 
state in their match-action tables and make comparisons across data blocks to see if they match a given query. However, there are several challenges in implementing queries on switches:

\parab{Function constraints:} There are limited operations we can run on switches (e.g. hashing, bit shifting, bit matching, etc). These are insufficient for queries which sometimes require string operations, and other arithmetic operations (e.g., multiplication, division, log) on numbers that are \mbox{not power-of-twos.}

\parab{Limited pipeline stages and ALUs:} 
Programmable switches use a pipeline of match action tables. The pipeline has a limited number of stages (e.g. 12-60) and a limited number of ALUs per stage. This means we can only do a limited number of computations at each stage (e.g. no more than ten comparisons in one stage for some switches). This is not enough for some queries which require many comparisons across entries (e.g. \distinct) or across many dimensions (e.g. \skyline).

\parab{Memory and bit constraints:} To reach high throughput and low latency, switches have a limited amount of on-chip memory (e.g. under 100MB of SRAM and up to 100K-300K TCAM entries) that is partitioned between stages. However, if we use switches to run queries, we have to store, compare, and group a large number of past entries that can easily go beyond the memory limit. Moreover, switches only parse a limited number of bits and transfer these bits across stages (e.g., 10-20 Bytes. Some queries may have more bits for the keys especially when queries are on multiple dimensions or long strings.

\section{\NAME design}\label{sec:design}
\label{sec:pruning}
\parab{The pruning abstraction:}
We introduce the {\em pruning} abstraction for {\em partially} offloading queries onto switches. 
Instead of offloading the complete query, we simply offload a critical part of a query. In this way, we best leverage the high throughput and low latency performance of switches while staying within their function and resource constraints. With pruning, the switch simply filters the data sent from the workers, but does not guarantee query completion. The master runs queries on the pruned dataset
and generates the same query result as if it had run the query with the original dataset. 
Formally, we define pruning as follows; Let $Q(D)$ denote the result (output) of query $Q$ when applied to an input data $D$. A pruning algorithm for $Q$, $A_Q$, gets $D$ produces $A_Q(D)\subseteq D$ such that $Q(A_Q(D))=Q(D)$.
That is, the algorithm computes a subset of the data such that the output of running $Q$ on the subset is equivalent to that of applying the query to the whole of $D$.

To make our description easier, for the rest of the paper, we focus on queries with one stage of workers and one master worker. We also assume a single switch between them. An example is a rack-scale query framework where all the workers are located in one rack and the top-of-rack switch runs \NAME pruning solutions. Our solutions can work with multiple stages of workers by having the switch prune data for each stage. We \mbox{discuss how to handle multiple switches in~\cref{sec:extensions}.}

\parab{\NAME architecture:}
\NAME can be easily integrated within Spark without affecting its normal workflow. Figure~\ref{fig:prototype} shows the \NAME design:

\begin{figure}
    \centering
    \includegraphics[width=1\linewidth]{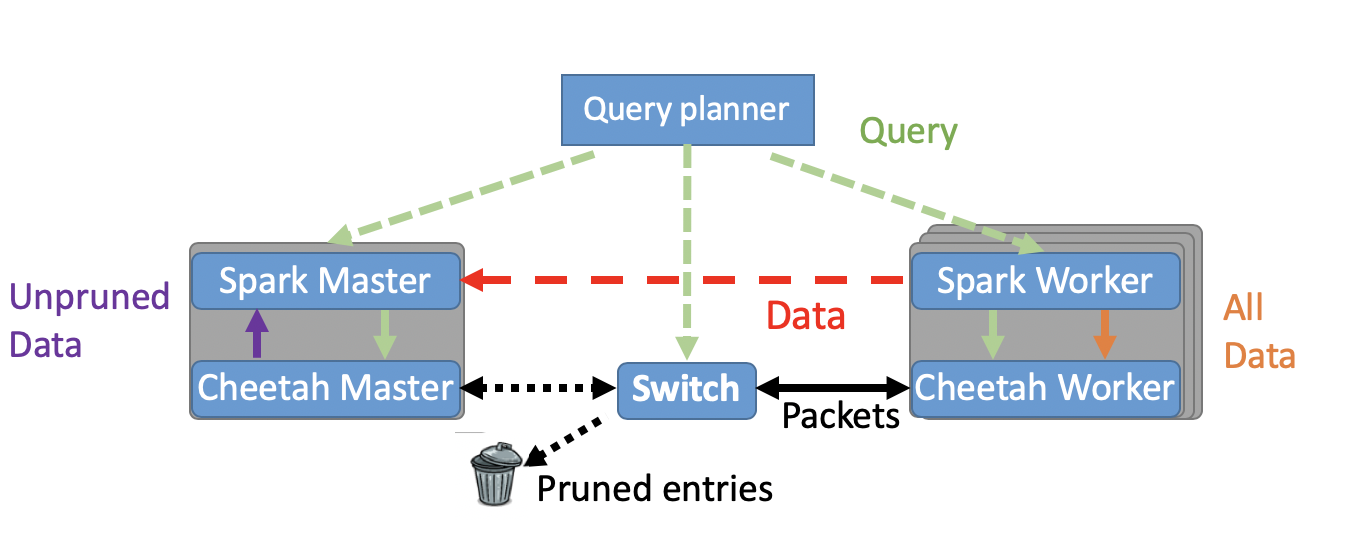}
    \caption{\small Cheetah Design.}
    \label{fig:prototype}
\end{figure}

\parabe{Query planner:}
The way users specify a query in Spark remains unchanged. For example, the query (e.g., SELECT * WHERE $c>\theta$) has three parameters: (1) the query type (filtering in this example), (2) The query parameters $\theta$, and (3) the relevant columns ($c$ in this example). 
In addition to distributing tasks to workers, the query planner sends (1) and (2) to the switch control plane which updates the switch memory accordingly. Once the master receives an ACK from the switch, which acknowledges that it is ready, it starts \mbox{the tasks at workers.}

\parabe{CWorkers:}
With \NAME, the workers do not need to run computationally intensive tasks on the data. Instead, we implement a \NAME module (CWorker), which intercepts the data flow at workers and sends the data directly through the switch. 
Therefore, \NAME reduces the processing time at the workers and partially offloads their tasks to the switch. The workers and the master only need to process the remaining part of the query. CWorkers also convert the related columns into packets that are readable by the switch. For example, if some entries have variable width or are excessively wide (e.g., a \distinct query on multiple columns), CWorkers may compute fingerprints before sending the data out. 

\parabe{\NAME Switch:}
\NAME switch is the core component of our system. We pre-compile all the common query algorithms at the switch. At runtime, when the switch receives a query and its parameters, it simply installs match-action rules based on the query specifications. Since most queries just need tens of rules, the rule installation takes less than 1 ms in our evaluation. According to these rules, the switch prunes incoming packets by leveraging ALUs, match-action tables, registers, and TCAM as explained in~\cref{sec:querypruningalgorithms}. The switch only forwards the remaining packets to the CMaster. The switch identifies incoming packets from workers based on pre-specified port numbers. This allows the switch to be fully compatible with other network functions and applications sharing the same network. Since the switch only adds acceleration functions by pruning the packets, the original query pipeline can work without the switch. If the switch fails, operators can simply reboot the switch with empty states or use a backup ToR switch. We also introduce a new communication protocol, as explained in detail in~\cref{sec:commProt}, that allows the workers to distinguish between pruned packets that are legitimately dropped and lost packets that should be retransmitted.

\parabe{CMaster:} At the master, we implement a \NAME module (CMaster) that converts the packets back to original Spark data format. The Spark master works in the same way with and without \NAME. It ``thinks'' that we are running the query on the pruned dataset rather than the original one and completes the operation. As the pruned dataset is much smaller, the Spark master takes less time to complete on \NAME{}. Many Spark queries adopt late materialization: Spark first runs queries on the metadata fields (i.e., those columns of an entry that the query conditions on) and then fetches all the requested columns for those entries that match the criteria. In this case, \NAME prunes data for the metadata query and does not modify the final fetches.

\begin{table}
{
  \centering
  \resizebox{.95 \columnwidth}{!}{
  \subfloat[\products]{%
                \begin{tabular}{lcc}
                \hline
                \name&\seller&\price\\\hline
                    Burger&McCheetah&4\\
                    Pizza&Papizza&7\\
                    Fries&McCheetah&2\\
                    Jello&JellyFish&5\\\\
                \hline
                \end{tabular}
  }

  \subfloat[\ratings]{%
                    \begin{tabular}{lcc}
                    \hline
                    \name&\taste&\texture\\\hline
                    Pizza&7&5\\
                    Cheetos&8&6\\
                    Jello&9&4\\
                    Burger&5&7\\
                    Fries&3&3\\
                    \hline
                    \end{tabular}
  }}
  }
  \caption{
  Running Database Example}\label{tbl:example}
  \vspace*{-5mm}
\end{table}

\section{Query Pruning Algorithms}\label{sec:querypruningalgorithms}
In this section, we explore the high-level primitives used for our query pruning algorithms.
An example of input tables which we will use to illustrate them is given in Table~\ref{tbl:example}. We also provide a table summarizing our algorithms, their parameters and pruning guarantee type in
\ifdefined\fullversion
Appendix~\ref{app:algTable}.
\else
the full version~\cite{fullVersion}.
\fi

\subsection{Handling Function Constraints}
Due to the switch's function and resource limitations, we cannot always prune a complete query.  
In such cases, \NAME automatically decomposes the query into two parts and prunes \emph{parts} of the query which are supported.

\parab{Example \#1: Filtering Query:}\\
Consider the common query in database of selecting entries matching a WHERE expression, for example:
\begin{lstlisting}[backgroundcolor = \color{lightgray!50},style=SQL, breaklines=true]
SELECT * FROM Ratings WHERE (taste > 5) OR (texture > 4 AND name LIKE e%s)
\end{lstlisting}
The switch may not be able to compute some expressions due to lack of function support (e.g., if it cannot evaluate \name LIKE e\%s) or may lack ALUs or memory to compute some functions. \NAME runs a subset of the predicates at the switch to prune the data and runs the remaining either at the \mbox{workers or the master.}

The \NAME query compiler decomposes the predicates into two parts: Consider a monotone Boolean formula $\phi=f(x_1,\dots,x_n, y_1,\ldots y_m)$ over binary variables $\set{x_i},\set{y_i}$ and assume that the predicates $\set{x_i}$ can be evaluated at the switch while $\set{y_i}$ cannot. The \NAME query compiler replaces each variable $y_i$ with a tautology (e.g., $(T\vee F)$) and applies standard reductions (e.g., modus ponens) to reduce the resulting expression. The resulting formula is computable at the switch and allow \NAME to prune packets. 

In our example, we transform the query into:
\begin{multline*}
\mbox{$(\taste > 5)$ OR $ (\texture > 4 $ AND ($T\vee F$) ) }   \\
\mbox{$\vdash\ (\taste > 5)$ OR $ (\texture > 4) $}.
\end{multline*}

Therefore, \NAME prunes entries that do not satisfy $(\taste > 5)$ OR $ (\texture > 4) $ and let the master node complete the operation by removing \inputRows{} for which $\neg(\taste > 5) \wedge (\texture > 4) \wedge \neg$(\name LIKE e\%s ). 

In other cases, \NAME uses workers to compute the predicates that cannot be evaluated on the switch. For instance, the CWorker can compute (\name LIKE e\%s ) and add the result as one of the values in the sent packet. This way the switch can complete the filtering execution as all predicate values are now known.

\NAME supports combined predicates by computing the basic predicates that they contain ($\taste > 5$ and $\texture > 4$ in this example) and then checking the condition based on the true/false result we obtain for each basic predicate. \NAME writes the values of the predicates as a bit vector and looks up the value in a truth table to decide whether to drop \mbox{or forward the packets.}

\subsection{Handling Stage/ALU Constraints}\label{sec:partitioning}
Switches have limited stages and limited ALUs per stage. Thus, we cannot compare the current \inputRow{} with a sufficiently large set of points. Fortunately, for many queries, we can partition the data into multiple \emph{\matrixRows{}} such that each \inputRow{} is only compared with those in its \matrixRow{}.
Depending on the query, the partitioning can be either randomized or hash-based as \mbox{in the following. }

\parab{Example \#2: \distinct Query:}\\
The \distinct query selects all the distinct values in an \inputCols{} subset, e.g.,
\begin{lstlisting}[backgroundcolor = \color{lightgray!50},style=SQL, breaklines=true]
SELECT DISTINCT seller FROM Products 
\end{lstlisting}
returns (Papizza, McCheetah, JellyFish). To prune \distinct queries, the idea is to store all past values in the switch. When the switch sees a new \inputRow{}, it checks if the new value matches any past values. If so, the switch prunes this \inputRow{}; if not, the switch forwards it to the master node. However, storing all \inputRows{} on the switch may take too much memory.

To reduce memory usage, an intuitive idea is to use Bloom filters (BFs)~\cite{bloom1970space}. However, BFs have false positives. For \distinct, this means that the switch may drop \inputRows{} even if their values have not appeared. Therefore, we need a data structure that ensures no false positives but can have false negatives. Caches match this goal. The master can then remove the false negatives to complete the execution. 

We propose to use a $d\times w$ \emph{matrix} in which we cache \inputRows{}. Every \matrixRow{} serves as a Least Recently Used (LRU) cache that stores the last $w$ \inputRows{} mapped to it. When an \inputRow{} arrives, we first \emph{hash} it to $\set{1,\ldots,d}$, so that the same \inputRow{} always maps to the same \matrixRow{}. \NAME then checks if the current \inputRow{}'s value appears in the \matrixRow{} and if so prunes the packet. To implement LRU, we also do a rolling replacement of the $w$ entries by replacing the first one with the new entry, the second with the first, etc.  By using multiple \matrixRows{} we reduce the number of per-packet comparisons to \mbox{allow implementation on switches.}

In 
\ifdefined\fullversion
Appendix~\ref{app:distinct}.
\else
the full version~\cite{fullVersion}.
\fi
we analyze the pruning ratio on random order streams. Intuitively, if row $i$ sees $D_i$ distinct values and each is compared with $w$ that are stored in the switch memory, then with probability at least $w/D_i$ we will prune every duplicate entry. For example, consider a stream that contains $D=15000$ distinct entries and we have $d=1000$ rows and $w=24$ columns. Then we are expected to prune 58\% of the \emph{duplicate entries} (i.e., entry values that have appeared previously).
\begin{theorem}
Consider a random order stream with $D > d\ln(200d)$ distinct entries\footnote{It's possible to optimize other cases, but this seems to be the common case.}. Our algorithm, configured with $d$ rows and $w$ columns is expected to prune at least $ 0.99\cdot \min\set{\frac{w \cdot d}{D\cdot e},1}$ fraction of the duplicate entries.
\end{theorem}

\subsection{Handling Memory Constraints}\label{sec:memory}
Due to switch memory constraints, we can only store a few past entries at switches. The key question is: how do we decide which data to store at switches that maximizes pruning rate without removing useful entries? We give a few examples below to show how we set thresholds (in the \topn query) or leverage sketches (in \join and \having) to \mbox{achieve these goals.}

\parab{Example \#3: \topn Query:}\\
Consider a \topn query (with an ORDER BY clause), in which we are required to output the $N$ \inputRows{} with the largest value in the queried \inputCol{};\footnote{In different systems this operation have different names; e.g., MySQL supports LIMIT while Oracle has ROWNUM.} e.g., 
\begin{lstlisting}[backgroundcolor = \color{lightgray!50},style=SQL, breaklines=true]
SELECT (*\bfseries TOP*) 3 name, texture FROM Ratings ORDER BY taste 
\end{lstlisting}
may return (Jello 4, Cheetos 6, Pizza 5). Pruning in a \topn query means that we may return a superset of the largest $N$ \inputRows{}. The intuitive solution is to store the $N$ largest values, one at each stage. We can then compare them and maintain a rolling minimum of the stages. However, when $N$ is much larger than the number of stages (say $N$ is $100$ or $1000$ compared to 10-20 stages), this approach does not work. 

Instead, we use a small number of threshold-based counters to enable pruning \topn query. The switch first computes the minimal value $t_0$ for the first $N$ entries. Afterward, the switch can safely filter out everything smaller than $t_0$. It then tries to increase the threshold by checking how many values larger than $t_1$ we observe, for some $t_1>t_0$. Once $N$ such entries were processed, we can start pruning entries smaller than $t_1$.
We can then continue with larger thresholds $t_2,t_3,\ldots,t_w$. We set the thresholds exponentially ($t_i=2^i\cdot t_0$) in case the first $N$ is much smaller than most in the data.  This power-of-two choice also makes it easy to implement in switch hardware. When using $w$ thresholds, our algorithm can may get pruning points smaller than $t_0\cdot 2^{w-1}$, if enough larger ones exist. 

\parab{Example \#4: \join Query:}\\
In a \join operation we combine two tables based on the target \inputCols{}.\footnote{We refer here to INNER JOIN, which is SQL's default. With slight modifications, \NAME can also prune LEFT/RIGHT OUTER joins.} For example, the query
\begin{lstlisting}[backgroundcolor = \color{lightgray!50},style=SQL, breaklines=true]
SELECT * FROM Products JOIN Ratings    ON Products.name = Ratings.name
\end{lstlisting}
gives\qquad{} 
                \begin{tabular}{lcccc}
                \hline
                \name&\seller&\price&\taste&\texture\\\hline
                    Burger&McCheetah&4&5&7\\
                    Pizza&Papizza&7&7&5\\
                    Fries&McCheetah&2&3&3\\
                    Jello&JellyFish&5&9&4\\
                \hline
                \end{tabular}\ \ .\\

In the example, we can save computation if the switch identifies that the key "Cheetos" did not appear in the \products table and prune it. To support \join, we propose to send the data through the switch with two passes. In the first pass, we use Bloom Filters~\cite{bloom1970space} to track observed keys. 
Specifically, consider joining tables A and B on \inputCol{} (or \inputCols{}) C. Initially, we allocate two empty Bloom filters $F_A, F_B$ to  approximately record the observed values
(e.g., \name) by using an \inputCol{} optimization to stream the values of C from both tables. Whenever a key $x$ from table A (or B) is processed on the switch, \NAME adds $x$ to $F_A$ (or $F_B$). Then, we start a second pass in which the switch prunes each packet $x$ from A (respectively, B) if $F_B$ ($F_A$) did not report a match. As Bloom Filters have no false negatives, we are guaranteed that \NAME does not prune any matched entry. In the case of \join, The false positives in Bloom Filters only affect the pruning ratio while the correctness is guaranteed.
Such a two pass strategy causes more network traffic but it significantly reduces the processing time at workers.

If the joined tables are of significantly different size, we can optimize the processing further. We first stream the small table without pruning while creating a Bloom filter for it. Since it is smaller, we do not lose much by not pruning and we can create a filter with significantly lower false positive rate. Then, we stream the large table while \mbox{pruning it with the filter.}

\parab{Example \#5: \having Query:}\\

\having runs a filtering operation on top of an aggregate function (e.g.,MIN/MAX/SUM/COUNT). For example, 
\begin{lstlisting}[backgroundcolor = \color{lightgray!50},style=SQL, breaklines=true]
SELECT seller FROM Products GROUP BY seller HAVING SUM(price) > 5
\end{lstlisting}
should return (McCheetah, Papizza). We first check the aggregate function on each incoming entry. For MAX and MIN, we simply maintain a counter with the current max and min value. If it is satisfied, we proceed to our Distinct solution (see \cref{sec:partitioning}) -- if it reports that the current key has not appeared before, we add it to the data structure and forward the entry; otherwise we prune it.

SUM and COUNT are more challenging because a single entry is not enough for concluding whether we should pass the entry. We leverage sketches to store the function values for different entries in a compact fashion. 
We choose Count-Min sketch instead of Count sketch or other algorithms because Count-Min is easy to implement at switches and it has one-sided error. That is, for {\em \having $f(x) > c$}, where $f$ is SUM or COUNT, Count-Min gives an estimator $\widetilde{f(z)}$ that satisfies $\widetilde{f(z)}\ge f(z)$.
Therefore, by pruning only if $\widetilde{f(z)}\le c$ we guarantee that every key $x$ for which $f(x) > c$ makes it to the master. Thus, the sketch estimation error only affects the pruning rate. After the sketch, the switches blocks all the traffic to the master. We then make a partial second pass (i.e., stream the data again), only for the keys requested by the master. That is, the master gets a superset of the keys that it should output and requests all \inputRows{} that belong to them. It can then compute the correct COUNT/MAX and remove the incorrect keys (whose true value is at most $c$).
We defer the support for $\mathit{SUM}$/$\mathit{COUNT} < c$ operations to future work.

\subsection{Projection for High-dimensional Data}
So far we mainly focus on database operations on one dimension. However, some queries depend on values of multiple dimensions (e.g., in \skyline). Due to the limited number of stages and memory at switches, it is not possible to store and compare each dimension. Therefore, we need to project the multiple dimensions to one value (i.e., a fingerprint). The normal way of projection is to use hashing, which is useful for comparing if an entry is equal to another (e.g., \distinct and \join). However, for other queries (e.g., \skyline), we may need to order multiple dimensions so we need a different projection strategy to preserve the ordering.

\parab{Example \#6: \skyline Query:}\\
The Skyline query~\cite{borzsony2001skyline} returns all points on the Pareto-curve of the $D$-dimensional input. 
Formally, a point $x$ is dominated by $y$ only if it is dominated on all dimensions, i.e., $\forall i\in\set{1,\ldots, D}: x_i\le y_i$. The goal of a \skyline query is to find all the points that are not dominated in a dataset.\footnote{For simplicity, we consider maximizing all dimensions. We can extend the solution to support minimizing all dimensions \mbox{with small modifications.}}
For example, the~query 
\begin{lstlisting}[backgroundcolor = \color{lightgray!50},style=SQL, breaklines=true]
SELECT name FROM Ratings        (*\bfseries SKYLINE*) (*\bfseries OF*) taste, texture 
\end{lstlisting}
should return (Cheetos, Jello, Burger). 

Because skyline relates to multiple dimensions, when we decide whether to store an incoming entry at the switch, we have to compare with all the stored entries because there is no strict ordering among them. For each entry, we have to compare all the dimensions to decide whether to replace it. But the switch does not support conditional write under multiple conditions in one stage. These constraints make it challenging to fit \skyline queries on the switch with a \mbox{limited number of stages.}

To address this challenge, we propose to project each high-dimensional \inputRow{} into a single numeric value.
We define a function $h:\mathbb R^D\to\mathbb R$ that gives a single \emph{score} to $D$-dimensional points. We require that $h$ is \emph{monotonically increasing} in all dimensions to ensure that if $x$ is dominated by $y$ then $h(x)\le h(y)$.  In contrast, $h(x)\le h(y)$ does not imply that $x$ is dominated by $y$. For example, we can define $h$ to be the sum or product of coordinates.

The switch stores a total of $w$ points in the switch. Each point $y_{(i)}, i=1..w$ takes two stages: one for $h(y_{(i)})$ and another for all the dimensions in $y_{(i)}$. When a new point $x$ arrives, for each $y_{(i)}$, the switch first checks if $h(x)>h(y_{(i)})$. If so, we replace $h(y_{(i)})$ and $y_{(i)}$ by $h(x)$ and $x$. Otherwise, we check whether $x$ is dominated by $y_{(i)}$ and if so mark $x$ for pruning \mbox{without changing the state.}
Note that here our replace decision is only based on a single comparison (and thus implementable at switches); our pruning decision is based on comparing all the dimensions but the switch only drops the packet at the end of the pipeline (not same stage action). 

If $x$ replaced some $y_{(i)}$, we put $y_{(i)}$ in the packet and continue the pipeline with the new values. We use a rolling minimum (according to the $h$ values) and get that the $w$ points stored in the switch are those with the highest $h$ value so far are among the true skyline.

The remaining question is which function $h$ should be. Product (i.e., $h_P(x)=\prod_{i=1}^D x_i$) is better than sum (i.e., $h_S(x)=\sum_{j=1}^D x_j$) because sum is biased towards the  the dimension with a larger range (Consider one dimension ranges between $0$ and $255$ and another between $0$ and $65535$). However, production is hard to implement on switches because it requires large values and multiplication. Instead, we use Approximate Product Heuristic (APH) which uses the monotonicity of the logarithm function to represent products as sum-of-logarithms and uses the switch TCAM and lookup tables to approximate the logarithm values (see more details in \ifdefined\fullversion
Appendix~\ref{app:skyline}).
\else
the full version~\cite{fullVersion}).
\fi 

\section{Pruning w/ Probabilistic Guarantees}
The previous section focuses on providing deterministic guarantees of pruning, which always ensures the correctness of the query results. Today, to improve query execution time, database systems sometimes adopt probabilistic guarantees (e.g.,~\cite{hu2019output}). This means that with a high probability (e.g., 99.99\%), we ensure that the output is \emph{exactly} as expected (i.e., no missing entries or extra entries). That is, $\Pr[Q(A_Q(D))\ne Q(D)]\le\delta$, where $Q$ is the query, $A_Q$ is the algorithm, and $D$ is the data, as in~\cref{sec:pruning}. Such probabilistic guarantees allow users to get query results much faster.\footnote{The master can check the extra entries before sending the results to users. Spark can also return the few missing entries at a later time.}

By relaxing to probabilistic guarantees, we can improve the pruning rate by leveraging randomized algorithms to select the entries to store at switches and adopting hashing to reduce the memory usage.

\parab{Example \#7: (Probabilistic) \topn Query:}\\
Our randomized \topn algorithm aims at, with a high probability,  returning a superset of the expected output (i.e., \emph{none} of the $N$ \outputRows{} is pruned).
\NAME \emph{randomly} partitions the \inputRows{} into \matrixRows{} as explained in~\cref{sec:partitioning}. Specifically, when an \inputRow{} arrives, we choose a random row for it in $\{1,\ldots,d\}$. In each \matrixRow{}, we track the largest $w$ \inputRows{} mapped to it using a rolling minimum. That is, the largest \inputRow{} in the \matrixRow{} is first, then the second, etc.
We choose to prune any \inputRow{} that was smaller than all $w$ \inputRows{} that were cached in its \matrixRow{}.
\NAME leverages the balls and bins framework to determine how to set the dimensions of the matrix ($d$ and $w$) given $N$, the goal probability $\delta$, and the resource constraints (the number of stages limits $w$ while the per-stage memory restricts $d$).

The algorithm is illustrated in Figure~\ref{fig:topn-example}.

        \begin{figure}[t]
    \centering
    {\includegraphics[width=0.99\linewidth]{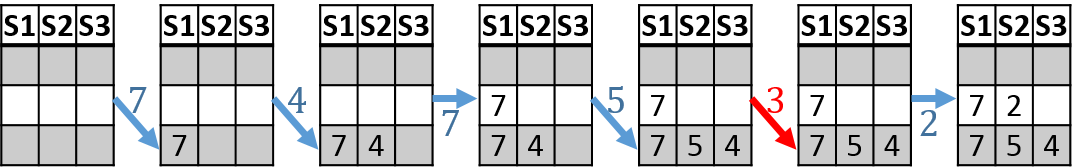}}
    \caption{\small \topn example on a stream (7,4,7,5,3,2). The \inputRow{} $3$ was mapped to the third \matrixRow{} and pruned as all stored values were larger. In contrast, $2$ was mapped to the second \matrixRow{} and is not pruned. The matrix dimensions are chosen so that with high probability none of the \topn \mbox{\inputRow{} is pruned.}
        ~\label{fig:topn-example}}
        \vspace*{-2mm}
\end{figure}

The proper $(w,d)$ configuration of the algorithm is quite delicate.
\ifdefined\fullversion
In Appendix~\ref{app:skyline},
\else
In the full version~\cite{fullVersion},
\fi
we analyze how to set $w$ given a constraint on $d$ or vice versa. We also show that we achieve the best pruning rate when the matrix size $d \cdot w$ is minimized (if there is no constraint on $d$ or $w$), thus optimizing the space and pruning simultaneously.

The goal of our algorithm is to ensure that with probability $1-\delta$, where $\delta$ is an error parameter set by the user, no more than $w$ \topn values are mapped into the same \matrixRow{}. In turn, this guarantees that the pruning operation is successful and that all \outputRows{} are not pruned. In the following, we assume that $d$ is given (this can be derived from the amount of per-stage memory available on the switch) and discuss how to set the number of \matrixCols{}.
To that end, we use $w\triangleq\floor{\frac{1.3\lnp{d/\delta}}{\lnp{\frac{d}{N\cdot e}\lnp{d/\delta}}}}$ \matrixCols{}. For example, if we wish to find the TOP 1000 with probability 99.99\% (and thus, $\delta=0.0001$) and have $d=600$ \matrixRows{} then we use $w=16$ \matrixCols{}. Having more space (larger $d$) reduces $w$; e.g., with $d=8000$ \matrixRows{} we need just $5$ \matrixCols{}. Having too few \matrixRows{} may require an excessive number of \matrixCols{} (e.g., $w=288$ \matrixCols{} are required for $d=200$) which may be infeasible due to the limited number of pipeline stages. 
Due to lack of space, the proof of the theorem appears in 
\ifdefined\fullversion
Appendix~\ref{app:topn}.
\else
the full version~\cite{fullVersion}.
\fi
\begin{theorem}\label{thm:topn}
Let $d,N\in\mathbb N, \delta > 0$ such that $d\ge N\cdot e / \ln\delta^{-1}$ and define $w\triangleq\floor{\frac{1.3\lnp{d/\delta}}{\lnp{\frac{d}{N\cdot e}\lnp{d/\delta}}}}$. Then \topn query succeeds with probability at least $1-\delta$.
\end{theorem}

In the worst case, if the input stream is monotonically increasing, the switch must pass all \inputRows{} to ensure correctness. 
In practice, streams are unlikely to be adversarial as the order in which they are stored is optimized for performance. 
To that end, we analyze the performance on random streams, or equivalently, arbitrary streams that arrive in a random order.
Going back to the above example, if we have $d=600$ \matrixRows{} on the switch and aim to find TOP $1000$ from a stream of $m=8M$ elements, our algorithm is expected to prune at least 99\% of the data. For a larger table of $m=100M$ entries our bound implies expected pruning of over 99.9\%. Observe that the logarithmic dependency on $m$ in the following theorem implies that our algorithm work better for larger datasets. 
The following theorem's proof is deferred to
\ifdefined\fullversion
Appendix~\ref{app:topn}.
\else
the full version~\cite{fullVersion}.
\fi
\begin{theorem}\label{thm:pruning}
Consider a random-order stream of $m$ elements and the \topn operation with algorithm parameters $d,w$ as discussed above. Then our algorithm prunes at least $\parentheses{m-w\cdot d\cdot{\lnp{\frac{m\cdot e}{w\cdot d}}}}$ of the $m$ elements in expectation.
\end{theorem}

\paragraph{Optimizing the Space and Pruning Rate}
The above analysis considers the number of \matrixRows{} $d$ as given and computes the optimal value for the number of \matrixCols{} $w$. 
However, unless one wishes to use the minimal number of \matrixCols{} possible for a given per-stage space constraint, we can simultaneously optimize the space and pruning rate.
To that end, observe that the required space for the algorithm is $\Theta(w\cdot d)$, while the pruning rate is monotonically decreasing in $w\cdot d$ as shown in Theorem~\ref{thm:pruning}.
Therefore, by minimizing the product $w\cdot d$ we optimize the algorithm in both aspects.
Next, we note that for a fixed error probability $\delta$ the value for $w$ is monotonically decreasing in $d$ as shown in Theorem~\ref{thm:topn}. Therefore we define $f(d)\triangleq w\cdot d \approx { \frac{d\cdot1.3\lnp{d/\delta}}{\lnp{\frac{d}{N\cdot e}\lnp{d/\delta}}}}$ and minimize it over the possible values of $d$.\footnote{This omits the flooring of $w$ as otherwise the function is not continuous. The actual optimum, which needs to be integral, will be either the minimum $d$ for that value or for $w$ that is off by $1$.}
The solution for this optimization is setting $d\triangleq {\delta}\cdot e^{W(N\cdot e^2/\delta)}$, where $W(\cdot)$ is the Lambert $W$ function defined as the inverse of $g(z)=ze^z$. For example, for finding TOP 1000 with probability 99.99\% we should use $d=481$ \matrixRows{} and $w=19$ \matrixCols{}, even if the per-stage space allows larger~$d$.

\parab{Example \#8: (Probabilistic) \distinct Query:}\\ 
Some \distinct queries run on multiple \inputCols{} or on variable-width fields that are too wide and exceed the number of bits that can be parsed from a packet. 
To reduce the bits, we use \emph{fingerprints},which are short hashes of all \inputCols{} that the query runs on. 

However, fingerprint collisions may cause the switch to prune \inputRows{} that have not appeared before and thus provide inaccurate output.\footnote{Note that for some other queries (e.g., \join), fingerprint collisions only affect the pruning rate, but not correctness.} Interestingly, not all collisions are harmful. This is because the \distinct algorithm hashes each entry into a row in $\{1,\ldots,d\}.$ Thus, a fingerprint collision between two \inputRows{} is bad \emph{only if they are in the same \matrixRow{}}. 

We prove the following bound on the required fingerprint length in
\ifdefined\fullversion
Appendix~\ref{app:distinct}.
\else
the full version~\cite{fullVersion}.
\fi
\begin{theorem}\label{thm:distinct}
Denote$$
\mathcal M\triangleq \begin{cases}
e\cdot D/d &\mbox{if $D>d\ln (2d/\delta)$}\\
e\cdot \ln (2d/\delta) &\mbox{if $d\cdot\ln\delta^{-1}/e \le D\le d\ln (2d/\delta)$}\\
\frac{1.3\lnp{2d/\delta}}{\lnp{\frac{d}{D\cdot e}\lnp{2d/\delta}}}&\mbox{otherwise}
\end{cases} ,
$$ where $D$ is the number of distinct items in the input.
Consider storing fingerprints of size $f=\ceil{\log(d\cdot \mathcal M^2/\delta)}$
bits. Then with probability $1-\delta$ there are no false positives and the distinct operation terminates successfully.
\end{theorem}
For example, if $d=1000$ and $\delta=0.01\%$, we can support up to $500M$ distinct elements using $64$-bits fingerprints \emph{regardless of the data size}. Further, this does not depend on \mbox{the value of $w$.}

The analysis leverages the balls and bins framework to derive bounds on the \emph{sum of square loads}, where each load is the number of distinct elements mapped into a \matrixRow{}. It then considers the number of distinct elements we can support without having same-row fingerprint collisions. For example, if $d=1000$ and the error probability $\delta=0.01\%$, we can support up to $500M$ distinct elements using $64$-bits fingerprints \emph{regardless of the data size}. Further, this does not depend on the value of $w$. 
We also provide a rigorous analysis of the  pruning rate in
Additionally, we analyze the expected pruning rate in random-order streams and show
\ifdefined\fullversion
, in Appendix~\ref{app:distinct},
\else
, in the full version~\cite{fullVersion},
\fi
that we can prune at least an $\Omega\parentheses{\frac{w \cdot d}{D}}$ fraction of the entries, where $D$ is the number of distinct \mbox{elements in the input.}

\section{Handling multiple queries}\label{sec:multiqueries}
\NAME supports the use case where the query is not known beforehand but only the set of queries (e.g., \distinct, \topn, and \join) we wish to accelerate. Alternatively, the workload may contain complex queries that combine several of our operations.
In this scenario, one alternative would be to reprogram the switch once a query arrives. However, this could take upwards of a minute and may not be better than to perform the query without \NAME{}. Instead, we concurrently pack the different queries that we wish to support on the switch data plane, splitting the ALU/memory resources between these. This limits the set of queries we can accommodate in parallel, but allow for \emph{interactive} query processing in a matter of seconds and without recompiling the switch.
Further, not all algorithms are heavy in the same type of resources. 
Some of our queries (e.g., \skyline) require many stages but few ALUs and only a little SRAM. In contrast, \join may use only a couple of stages while requiring most of the SRAM in them. These differences enable \NAME to efficiently pack algorithms on the same stages.

At the switch, all queries will be performed on the incoming data giving us a prune/no-prune bit for each query. Then we have a single pipeline stage that selects the bit relevant to the current query. We fit multiple queries by repurposing the functionality of ALU results and stages. We evaluate one such combined query in figure~\ref{fig:benchmark}. Query A is a filtering query and query B is a SUM + group by query. To prune the filtering query, we only use a single ALU and 32 bits of stage memory (1 index of a 32 bit register) in a stage. We use the remaining ALUs in the same stage to compute 1) hash values and 2) sums required for query B as discussed in our pruning algorithms. We also use the remaining stage memory in that same stage to store SUM results ensuring the additional filter query has no impact on the performance of our group by query. 

In more extreme examples, where the number of computation operations required exceeds the ALU count on the switch, it is still possible to fit a set of queries by reusing ALUs and registers for queries with similar processing. As an example, an ALU that does comparisons for filtering queries can be reconfigured using control plane rules to work as part of the comparator of a \topn or \having query. We can also use a single stage for more than one task by partitioning its memory e.g dedicating part of to fingerprinting for \distinct and another part to store \skyline prune points.

\begin{table*}[t!]
    \centering
    \resizebox{1\linewidth}{!}{
    \begin{tabular}{|ll|c|c|c|c|c|
    }\hline
      \multicolumn{2}{|c|}{\textbf{Algorithm}} & \textbf{Defaults} & \textbf{\#stages} & \textbf{\#ALUs} & \textbf{SRAM} & \textbf{\#TCAM}
      \\\hline\hline
      \multirow{2}{*}{\distinct}{}&FIFO$^{*}$&\multirow{2}{*}{$w=2, d=4096$}&$\ceil{w/A}$&$w$ 
      & \multirow{2}{*}{$(d\cdot w)\times$ 64b}& \multirow{2}{*}{$0$}
      \\\hhline{~-~--}
      &LRU&&$w$&$w$&&
      \\\hline\hline
      \multirow{2}{*}{\skyline{}}&SUM&\multirow{2}{*}{$D=2, w=10$}&\multirow{1}{*}{$\ceil{\log_2 D}+2w$}&\multirow{2}{*}{$2\ceil{\log_2 D}-1+w(D+1)$}&$w(D+1)\times$ 64b&0
      \\\hhline{~-~-~--}
      &APH&&\multirow{1}{*}{$\ceil{\log_2 D}+2(w+1)$}&&$w(D+1)\times$ 64b + $2^{16}\times$ 32b&$64\cdot D$
      \\\hline\hline
      \multirow{2}{*}{\topn{}}&Det&$N=250, w=4$&$w+1$&$w+1$&$(w+1)\times$ 64b&\multirow{2}{*}{$0$}
      \\\hhline{~-----}
      &Rand&$N=250,w=4,d=4096$&$w$&$w$&$(d\cdot w)\times$ 64b&
     \\\hline\hline
      \multicolumn{2}{|c|}{\groupby{}}&$w=8$&$w$&$w$&$d\cdot w\times$ 64b&0
      \\\hline\hline
      \multirow{2}{*}{\join{}}&BF$^{*}$&\multirow{2}{*}{$M=4$MB, $H=3$}&$2$&$H$&$M$&\multirow{2}{*}{$0$}\\\hhline{~-~---}
      &RBF&&$1$&$1$&$M+\binom{64}{H}\times$ 64b&
      \\\hline\hline
      \multicolumn{2}{|c|}{\having{}}&$w=1024, d=3$&$\ceil{d/A}$&$d$&$(d\cdot w)\times$ 64b&0\\
       \hline
    \end{tabular}
    }
    \caption{Summary of the resource consumption of our algorithms. Here, $A$ is the number of ALUs per stage on the switch. The algorithms denoted by (*) assume that \emph{same-stage} ALUs can access the same memory space. For \skyline{} the above assumes that the dimension satisfies $D\le A$.}
    \label{tab:alg2hardware}
\end{table*}

\section{Implementation}\label{sec:impl}
\subsection{\NAME prototype}
We built the \NAME{} dataplane along with in-network pruning using a Barefoot Tofino~\cite{tofino} switch and P4~\cite{P4Spec}. Each query requires between 10 to 20 control plane rules excluding the rules needed for TCP/IP routing and forwarding. Any of the Big Data benchmark workloads can be configured using less than 100 control plane rules. We also developed a DPDK-based Cheetah end-host service using \mbox{about 3500 lines of C.}

We deploy five Spark workers along with an instance of CWorker connected to the switch via DPDK-compliant 40G Mellanox NICs. We restrict the NIC bandwidth to 10G and 20G for our evaluation. All workers have two CPU cores and 4 GB of memory. The CWorker sends data to the master via UDP at a rate of $\sim$10 million packets per second (i.e a throughput of $\sim$ 5.1 Gbps since the minimum ethernet frame is $64$ bytes) with one entry per packet. We use optimized tasks for Spark for a fair comparison. We also mount a linux tmpfs RAM disk on workers to store the dataset partitions allowing Spark to take advantage of its main-memory optimized query plans.

\begin{figure*}
    \centering
    \includegraphics[width=\textwidth]{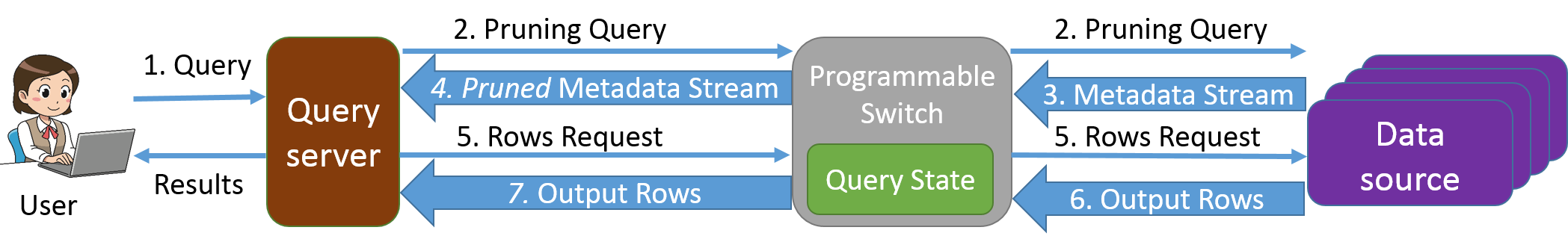}
    \caption{The control flow of Apache Spark and Cheetah}
    \label{fig:paperInAFig}
\end{figure*}

Spark optimizes the completion time by minimizing data movement. In addition to running tasks on workers to reduce the volume sent to the master, Spark compresses the data and packs multiple \inputRows{} in each packet (often, the maximum allowed by the network MTU).
In contrast, \NAME must send the data uncompressed while packing only a small number of entries in each packet. 
\ifdefined\fullversion
As additional optimization, Spark may use \emph{late materialization}. In late materialization, Spark only sends a metadata stream that contains the \inputCols{} that are specified by the query on the first pass. Next, the master decides which \inputRows{} should appear in the final output and requests them from the workers. For example, in our \join query example from~\cref{sec:memory}, Spark may only send the \name column in the first pass. The master then computes the set of keys that are in the intersection of the tables and asks for the relevant \inputRows{}. 
Late materialization is also illustrated in Figure~\ref{fig:paperInAFig}.

For plans that allow late materialization, the switch pruning only occurs in the first round of data movement when the columns relevant to the query are loaded, and does not interfere with the late materialization stage. Hence data can still be transferred in a compressed format during late materialization.
\else
Spark also leverages an optimization called \emph{late materialization}~\cite{latematerialization} in which only a metadata stream is sent in the first stage, and the entire \inputRows{} are requested by the master once it computes which tuples are part of the result. We expand on how \NAME supports late materialization in the full version~\cite{fullVersion}.
\fi

Currently, our prototype includes the \distinct, \skyline, \topn, \groupby, \join, and filtering queries. We also support combining these queries and running them in parallel without reprogramming the switch.

\noindent\textbf{\NAME Modules}
We create two background services that communicate with PySpark called \NAME Master (CMaster) and \NAME Worker (CWorker) running on the same servers that run the Spark Master and Spark Workers respectively. 
The CMaster bypasses a supported PySpark query and instead sends a control message to all CWorkers providing the dataset and the columns relevant to the query.  The CWorker packs the necessary columns into UDP packets with Cheetah headers and sends them via Intel DPDK~\cite{dpdk}. Our experiments show a CWorker can generate over $\approx$12 million packets per second when the \mbox{link and NIC are not a bottleneck.}

In parallel, the master server also communicates with the switch to install the control plane rules relevant to the query. The switch parses and prunes some of the packets it processes. The remaining packets are received at the master using an Intel DPDK packet memory buffer, are parsed, and copied into userspace. The remaining processing for the query is implemented in C. The \NAME master continues processing all entries it receives from the switch until receiving FINs from all CWorkers, indicating that data transmission is done. Finally, the CMaster sends the final set of row ids and values to the Python PySpark script (or shell). 

\noindent\textbf{Switch logic}
We use Barefoot Tofino~\cite{tofino} and implement the pruning algorithms (see~\cref{sec:design}) in the P4 language.  The switch parses the header and extracts the values which then proceed to the algorithm processing. The switch then decides if to prune each packet or forward it to the master. It also participates in our reliability protocol, which takes two pipeline stages on the hardware switch.

\noindent\textbf{Resource Overheads}
All our algorithms are parametric and can be configured for a wide range of resources. We summarize the hardware resource requirements of the algorithms in Table~\ref{tab:alg2hardware} and expand on how this is calculated in the full version of the paper~\cite{fullVersion}.
\subsection{Communication protocol}\label{sec:commProt}
\begin{figure}
        {
        \includegraphics[width =1.0\columnwidth]
            {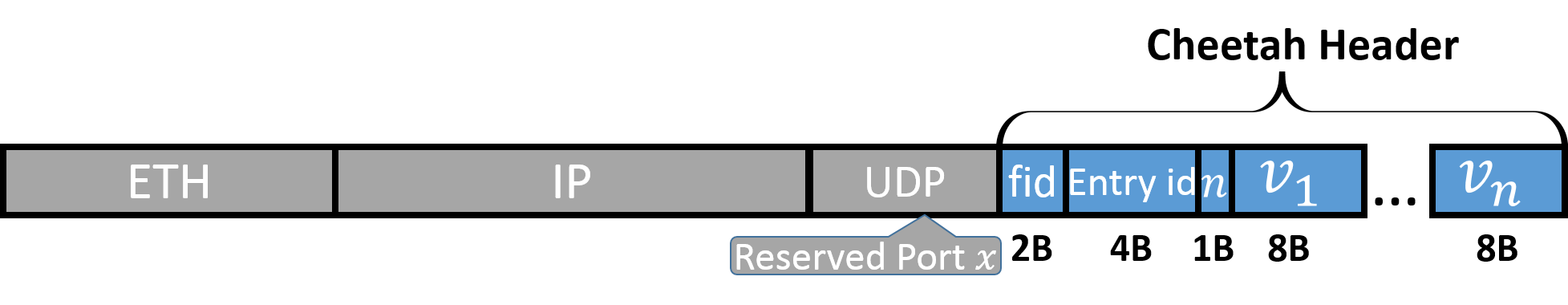} 
            \includegraphics[width =0.83\columnwidth]
            {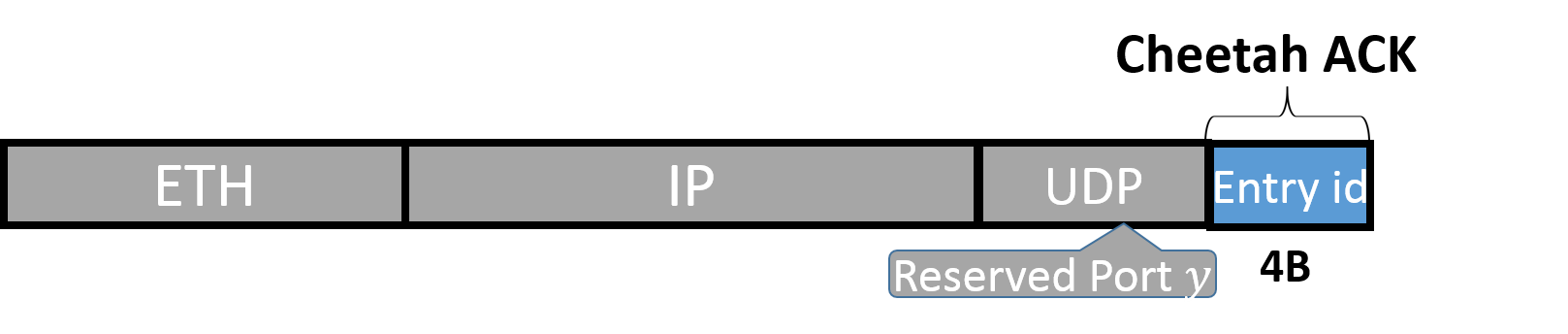}
        }
    \caption{\NAME{} packet and ACK format. The packets encode the flow and entry identifiers in addition to the relevant values.  }\label{fig:header}
\end{figure}

\parab{Query and response format:} 
For communication between the CMaster node and CWorkers (passing through and getting pruned by the switch), we implement a reliable transmission protocol built on top of UDP. 
Each message contains the \inputRow{} identifier along with the relevant column values or hashes.
Our protocol runs on a separate port from Spark and uses a separate header. It also does not use Spark's serialization implementation. 
Its channel is completely decoupled from and transparent to the ordinary communication between the Spark master and worker nodes. Our packet and header formats appear in Figure~\ref{fig:header}. We support variable length headers to allow the different number of columns / column-widths (e.g., \topn has one value per entry while \join/\groupby{} have two or more). The number of values is specified in an 8-bits field ($n$).
The flow ID (fid) field is needed when processing multiple datasets \mbox{and/or queries concurrently.}

For simplicity, we store one entry on each packet; We discuss how to handle multiple entries per packet in Section~\ref{sec:extensions}. 

\parab{Reliability protocol:}
We use UDP connections for CWorkers to send metadata responses to CMasters to ensure low latency. However, we need to add a reliability protocol on top of UDP to ensure the correctness of query results.

The key challenge is that we cannot simply maintain a sequence number at CWorkers and identify lost packets at CMasters because the switch prunes some packets. Thus, we need the switch to participate in the reliability protocol and acks the pruned packets to distinguish them from unintentional packet losses.

Each worker uses the \inputRow{} identifiers also as packet sequence numbers.
It also maintains a timer for every non-ACKed packet and retransmits it if no ACK arrives on time. The master simply acks every packet it receives.
For each fid, the switch maintains the sequence number $X$ of the last packet it processed, regardless of whether it was pruned. When a packet with SEQ $Y$ arrives at the switch the taken action depends on the relation \mbox{between $X$ and $Y$.}

If $Y=X+1$, the switch processes the packet, increments $X$, and decides whether to prune or forward the packet. If the switch prunes the packet, it sends an ACK($Y$) message to the worker. Otherwise, the master which receives the packet sends the ACK.
If $Y\le X$, this is a retransmitted packet that the switch processed before. Thus, the switch forwards the packet without processing it.
If $Y>X+1$, due to a previous packet $X+1$ was lost before reaching the switch, the switch drops the packet and waits for $X+1$ to be retransmitted.

This protocol guarantees that all the packets either reach the master or gets pruned by the switch. 
Importantly, the protocol maintains the correctness of the execution \emph{even if some pruned packets are lost and the retransmissions make it to the master}. The reason is that all our algorithms have the property that any superset of the data the switch chooses not to prune results in the same output. For example, in a \distinct query, if an \inputRow{} is pruned but its retransmission reaches the master, it can simply remove it.

\section{Evaluation}\label{sec:eval}

\begin{figure*}\label{fig:benchmarkAlgorithms}\label{fig:benchmarkQueries}
\includegraphics[width =0.85\textwidth]
            {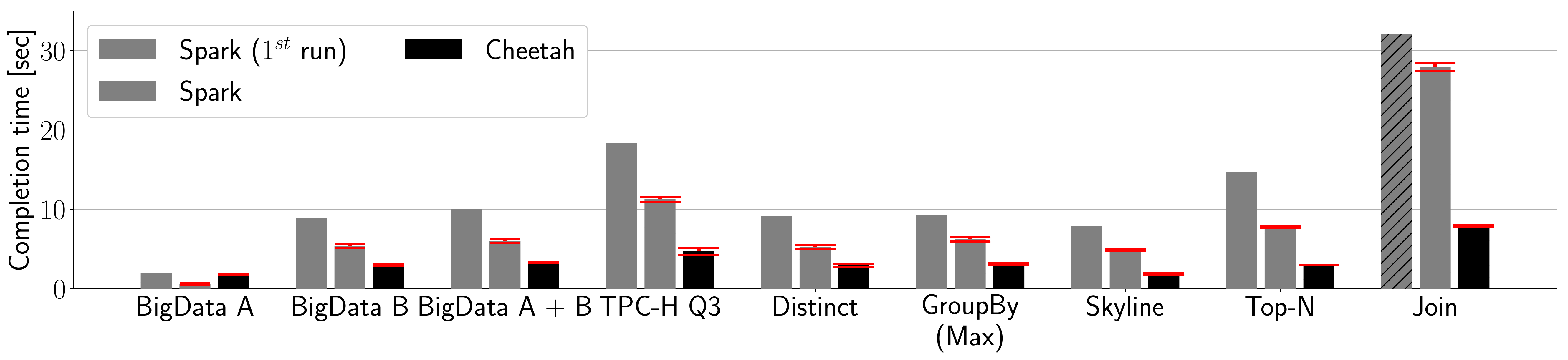}
    \caption{\label{fig:benchmark}\sigmodrev{A comparison of \NAME and Spark in terms of completion time on the Big Data benchmark for the benchmark queries (first four) and the other queries supported by \NAME.
    }
    }

\end{figure*}

We perform test-bed experiments and simulations. Our test-bed experiments show that \NAME has $40-200\%$ improvement in query completion time over Spark. Our simulations show that \NAME achieves a high pruning ratio with a modest amount of switch resources.

\subsection{Evaluation setup}\label{sec:eval_setup}
\parab{Benchmarks:}
Our test-bed experiments use the Big Data~\cite{pavlo} and TPC-H~\cite{tpch} benchmarks.
From the Big Data benchmark, we run queries A (filtering)\footnote{As the data is nearly sorted on the filtered column, we run the query on a random permutation of the table.}, B (we offload group-by), and A+B (both A and B executed sequentially). 

For TPC-H, we run query 3 which consists of two join operations, three filtering operations, a group-by, and a top N. We also evaluate each algorithm separately using a dedicated query on the Big Data benchmark's tables. All queries are shown in the full version of the paper~\cite{fullVersion}.

\subsection{Testbed experiments}

We run the BigData benchmark on a local cluster with five workers and one master, all of which are connected directly to the switch.
Our sample contains 31.7 million rows for the uservisits table and 18 million rows for the rankings table. We run the TPC-H benchmark at its default scale with one worker and one master. \NAME{} offloads the join part of the TPC-H because it takes 67\% of the query time and is the most effective use of switch resources.

\subsubsection{Benchmark Performance } 

Figure~\ref{fig:benchmark} shows that \NAME{} decreases the completion time by $64-75\%$ in both BigData B, BigData A$+$B, and TPC-H Query 3 compared to Spark's $1^{st}$ run and $47-58\%$ compared to subsequent runs.  Spark's subsequent runs are faster than $1^{st}$ run because Spark indexes the dataset based on the given workload after the  $1^{st}$ run. \NAME{} reduces the completion time by $40-72\%$ for other database operations such as distinct, groupby, Skyline, TopN, and Join. \NAME{} improves performance on these computation intensive aggregation queries because it reduces the expensive task computation Spark runs at the workers by offloading it to the switch's data plane instead.

BigData A (filtering) does not have a high computation overhead. Hence \NAME{} has performance comparable to Spark's 1st run but worse than Spark's subsequent runs. This is because \NAME{} has the extra overhead of serializing data at the workers to allow processing at the switches. This serialization adds more latency than the time saved by switch pruning. \NAME{} performs the combined query A $+$ B faster than the sum of individual completion times. This is because it pipelines the pre-processing of columns for the combined query resulting in faster serialization at CWorker.

\subsubsection{Effect of Data Scale and Number of Workers}\label{sec:scaling}\ 
In Figure~\ref{fig:vsNrEntries}, we vary the number of entries per partition (worker) while keeping the total number of entries fixed. Not only is \NAME{} is quicker than Spark, the gap widens as the data scale grows. Therefore \NAME{} may also offer better scalability for large datasets. Figure~\ref{fig:vsNrWorkers} shows the performance when fixing the overall amount of entries and varying the number of workers. \NAME{} improves Spark by about the same factor with a different number of partitions. In both these experiments, we ignore the completion time of Spark's first run on the query and only show subsequent runs (which perform better due to caching / indexing and JIT compilation effects ~\cite{dataanalyticsbottleneck,sparksql}). 
\begin{figure}
    \subfloat[Varying partition size]{\label{fig:vsNrEntries}\includegraphics[width =0.49\linewidth]
    {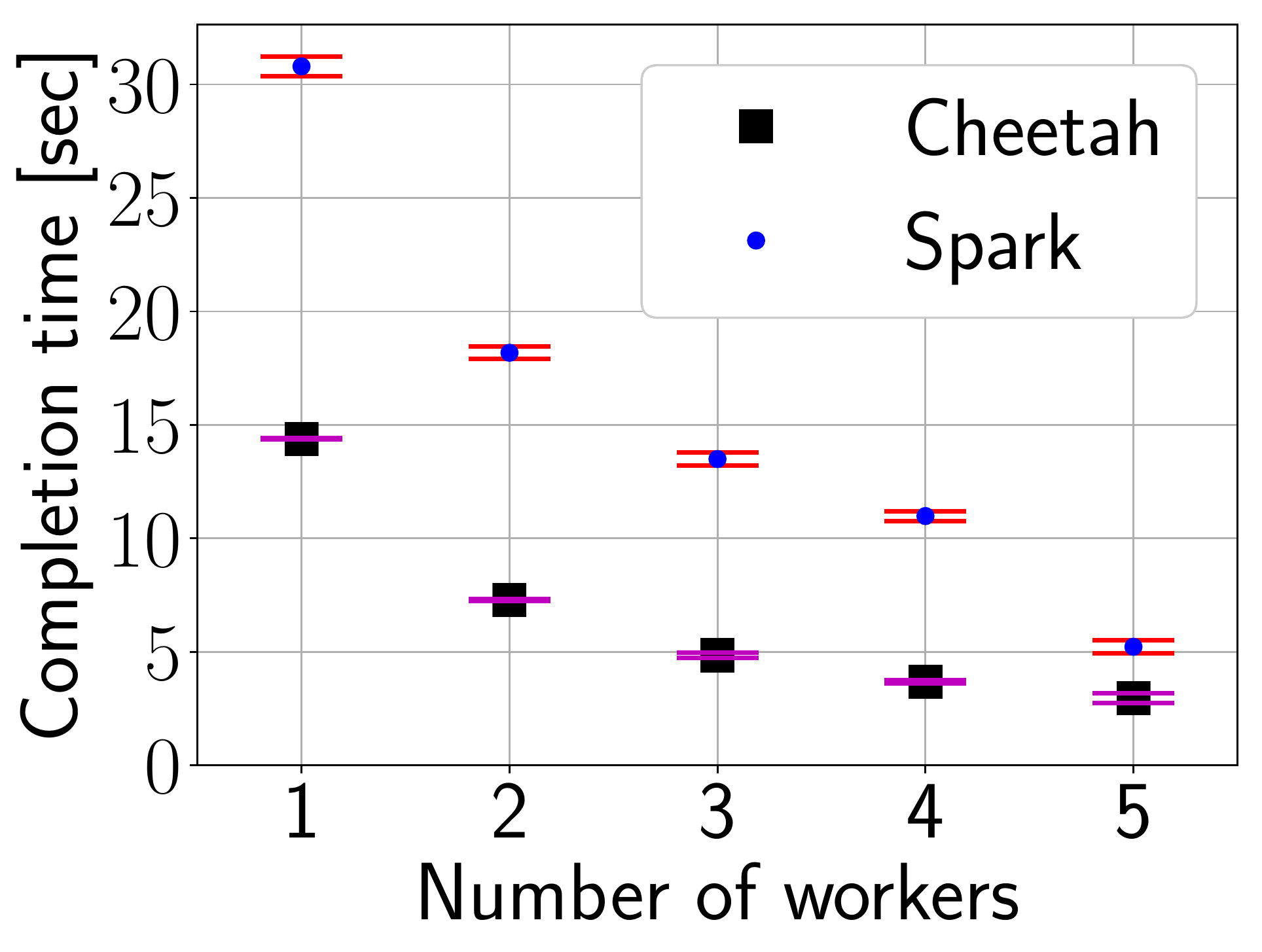}}
    \subfloat[Varying partition count]{\label{fig:vsNrWorkers}\includegraphics[width =0.49\linewidth]
    {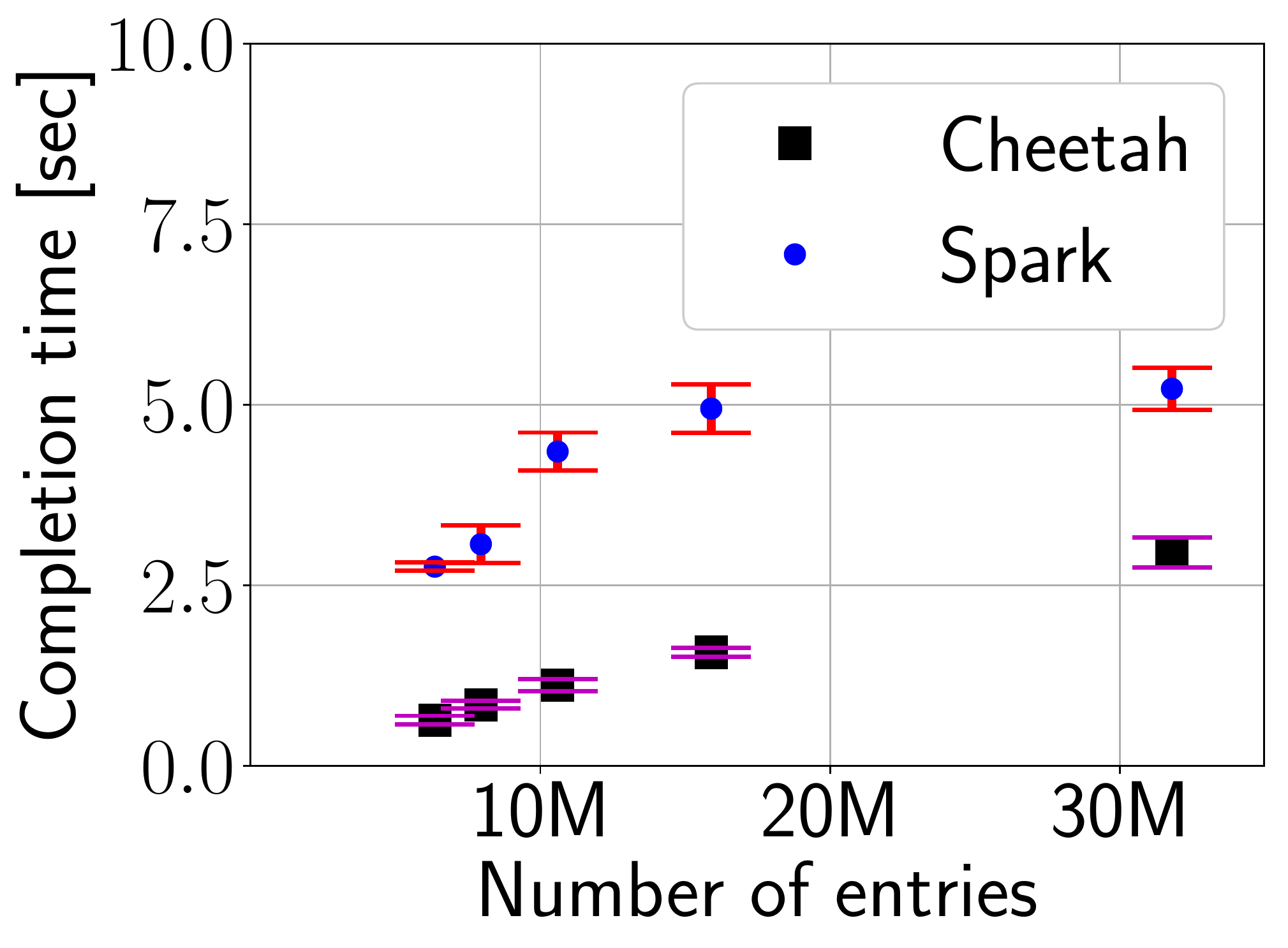}} 
    \vspace*{-2mm}
    \caption{\label{fig:scalability}\small The performance of \NAME{} vs. Spark SQL on \distinct query.
    } 
    \vspace*{-5mm}
\end{figure}

\subsubsection{Effect of Network Rate}\ 
Unlike Spark, which is often bottlenecked by computation~\cite{dataanalyticsbottleneck}, \NAME{} is mainly limited by the network when using a 10G NIC limit. We run \NAME{} using a 20G NIC limit and show a breakdown analysis of where each system spends more time.
 Figure~\ref{fig:breakdown} illustrates how \NAME{} diminishes the time spent at the worker at the expense of more sending time and longer processing at the master. The computation here for \NAME{} is done entirely at the master server, with the workers just serializing packets to send them over the network. When the speed is increased to 20G, the completion time of \NAME improves by nearly 2x, meaning that the network is the bottleneck. 
Similarly to Section~\ref{sec:scaling}, we discard Spark's first run.

\begin{figure}
    \centering
    \includegraphics[width=0.47\textwidth]{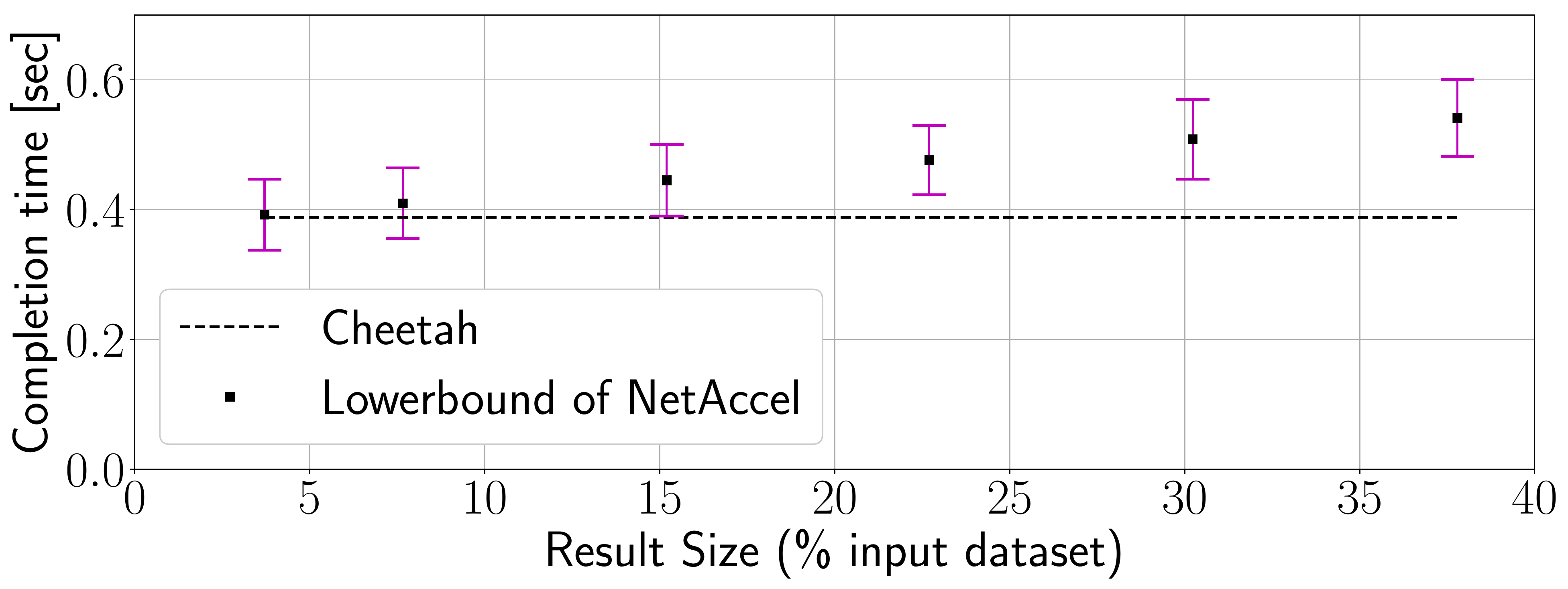}
    \caption{\small Overhead of moving results from the switch dataplane to the master server via packet draining on TPC-H Q3's order key join. We vary the result size by changing filter ranges in the query.}
    \label{fig:netaccel_result_size}
\end{figure}

\subsubsection{Comparison with NetAccel~\cite{lerner2019case}}\label{sec:netaccel} \; NetAccel is a recent system that offloads entire queries to switches. Since the switch data plane limitations may not allow a complete computation of queries, NetAccel \emph{overflows} some of the data to the switch's CPU. At the end of the execution, NetAccel drains the output (which is on the switch) and provides it to the user.
\NAME{}'s fundamental difference is that it does not aim to complete the execution on the switch and only prunes to reduce data size. As a result, \NAME{} is not limited to queries whose execution can be done on the switch (NetAccel only supports join and group by), does not need to read the output from the switch (thereby saving latency), and can efficiently pipeline query execution.
 
As NetAccel stores results at switches, it must drain the output to complete the execution. This process adds latency, as shown in figure~\ref{fig:netaccel_result_size}. We note that NetAccel's code is not publically available, and the depicted results are a lower bound obtained by measuring the time it takes to read the output from the switch. That is, this lower bound represents an ideal case where there are enough resources in the data plane for the entire execution and no packet is overflowed to the CPU. We also assume that NetAccel's pruning rate is as high as \NAME{}'s. Moreover, the query engine cannot pipeline partial results onto the next operation in the workload if \mbox{it stored in the switch.}

\begin{figure}
        \includegraphics[width =.3\textwidth]
            {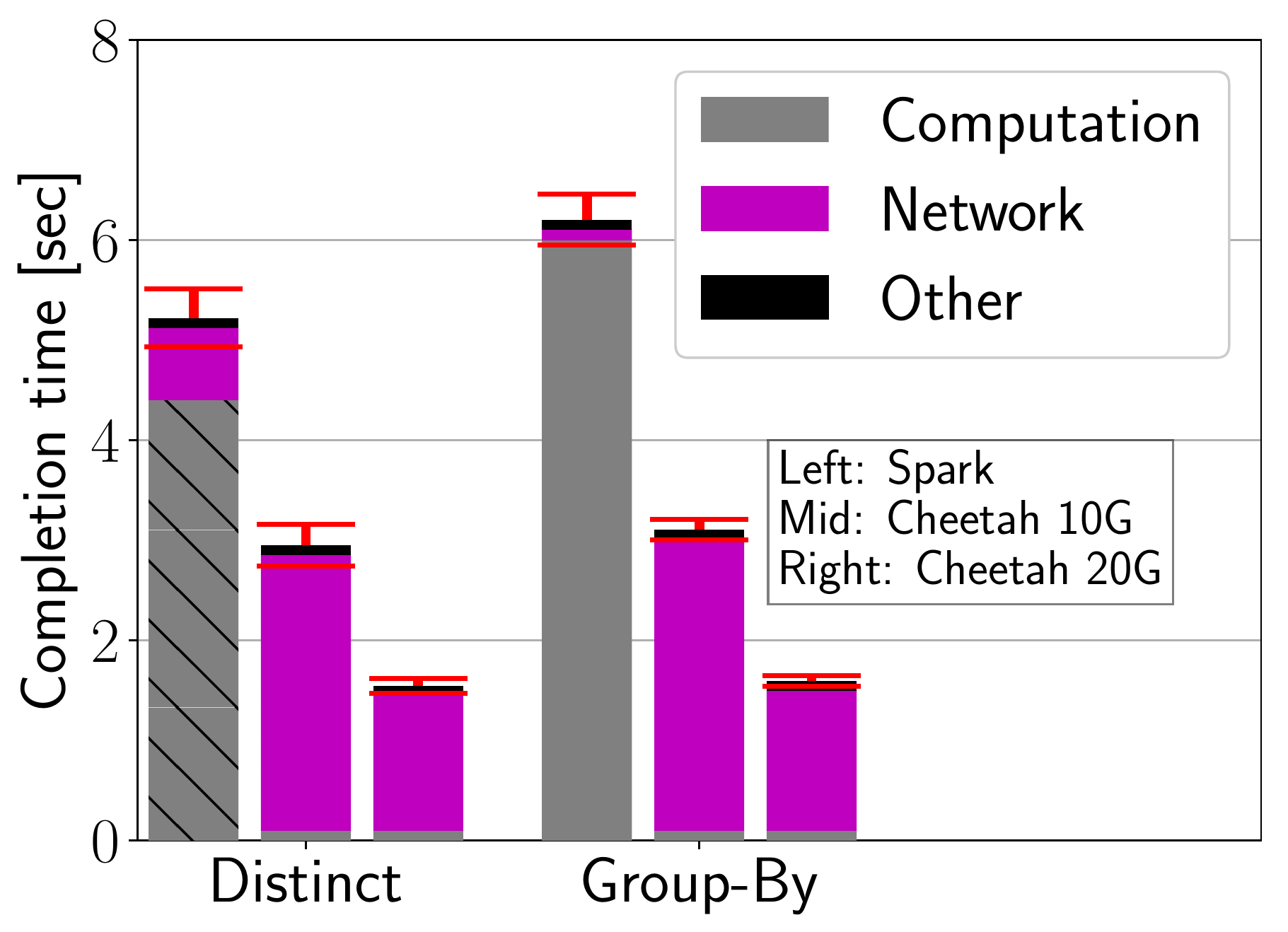} 
    \caption{{\small breakdown of Spark and \NAME{}'s delay for different network rates. Spark's bottleneck is not the network and it does not improve when using a faster NIC.}}
    \label{fig:breakdown}
\end{figure}
\begin{figure}
    \centering{
        {\includegraphics[width =0.8\columnwidth]
            {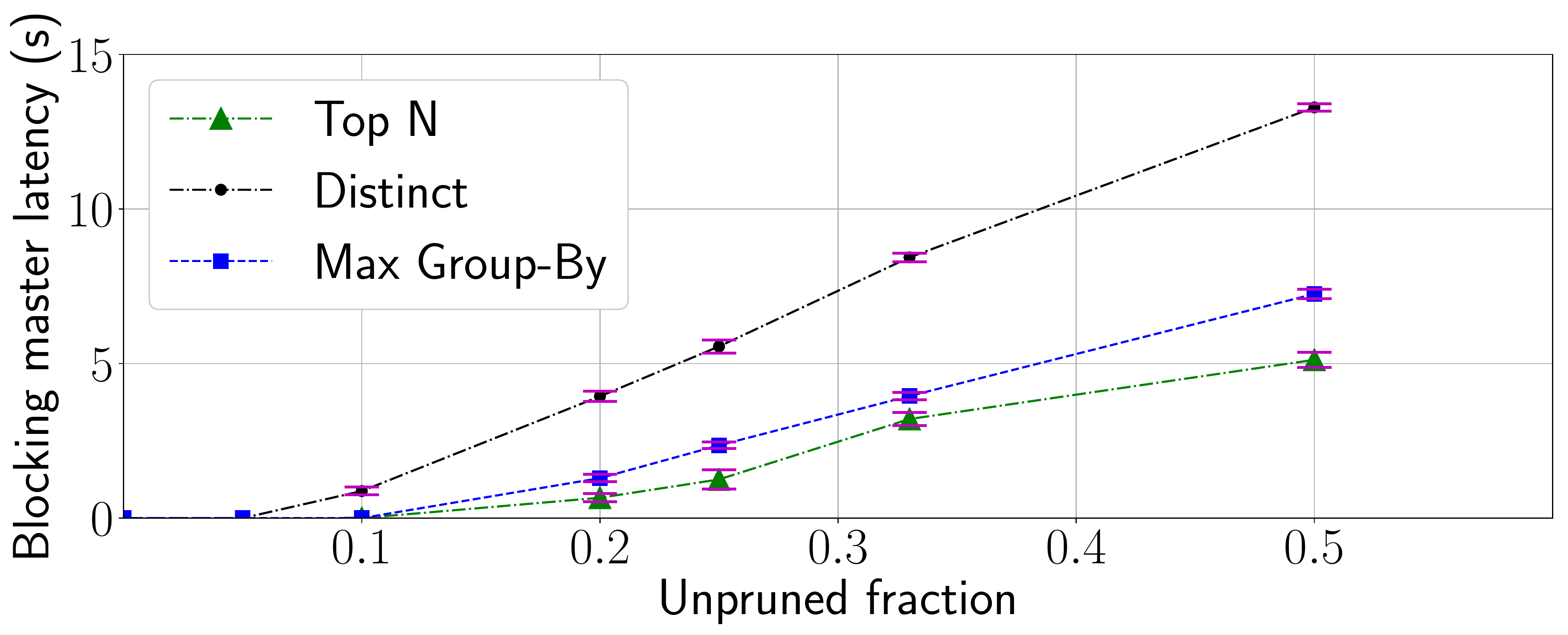}} 
    }
    \caption{\label{fig:MasterTime}\small The time it takes the Master to complete \distinct and max-\groupby queries for a given pruning rate.}
    \vspace*{-4mm} 
\end{figure}

\ifdefined\fullversion
\begin{figure*}[t]
        \subfloat[\distinct ($w=2$)]{\label{fig:distinct}\includegraphics[width =0.32\textwidth]
            {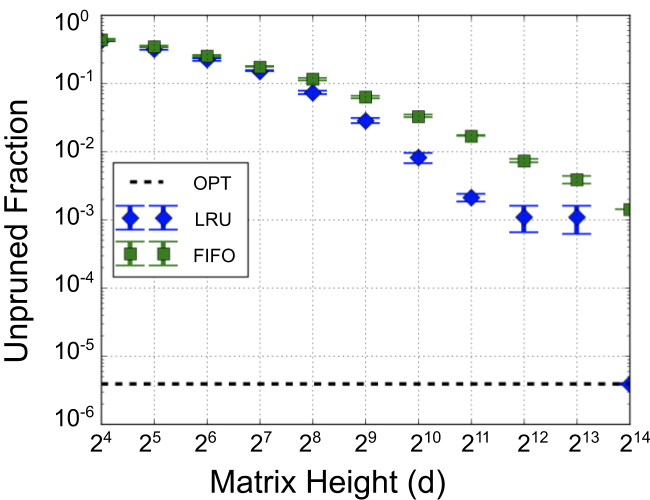}} 
        \subfloat[\skyline]{\label{fig:skyline}\includegraphics[width =0.32\textwidth]
            {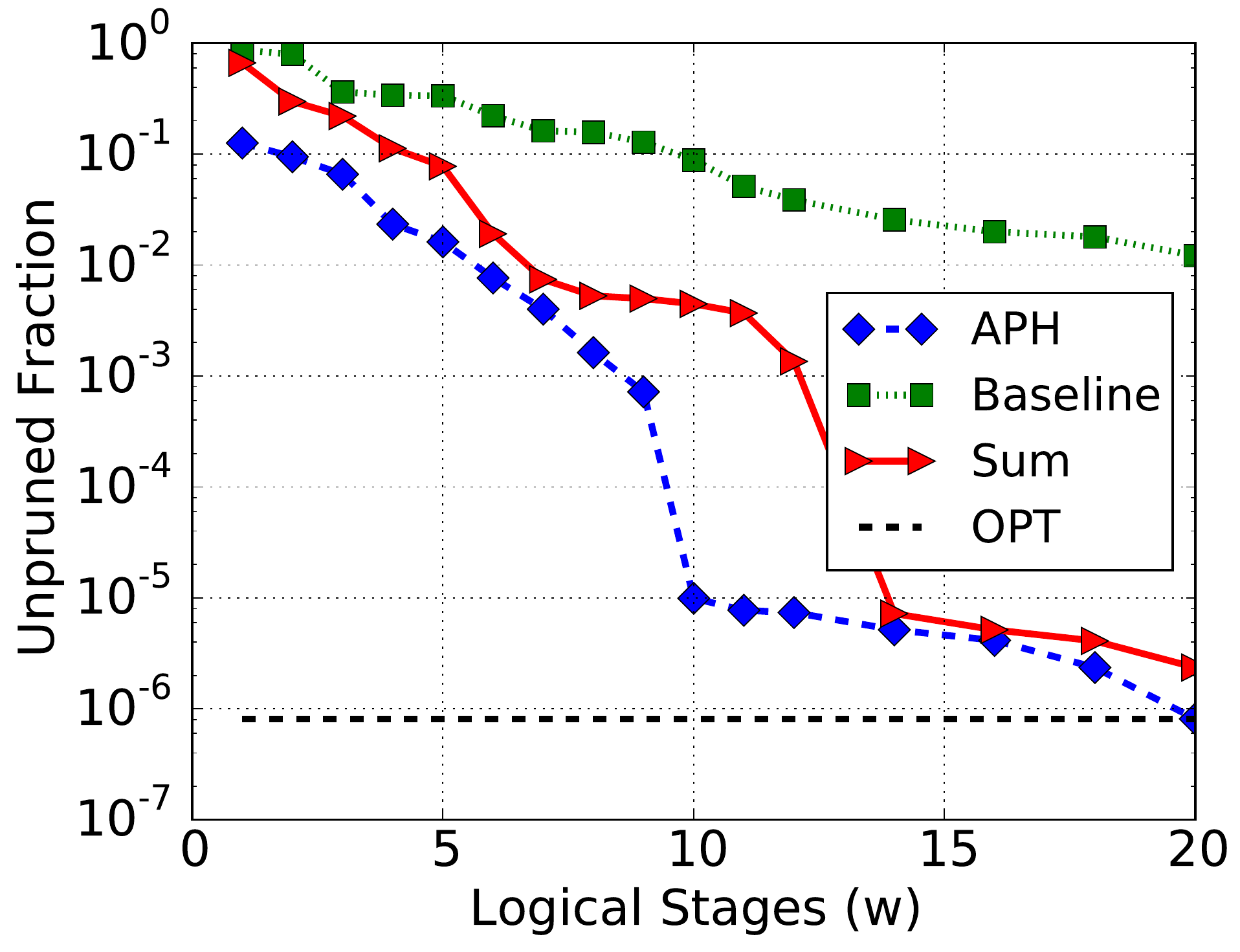}} 
        \subfloat[\topn ($d=4096$)]{\label{fig:topn}\includegraphics[width =0.32\textwidth]
            {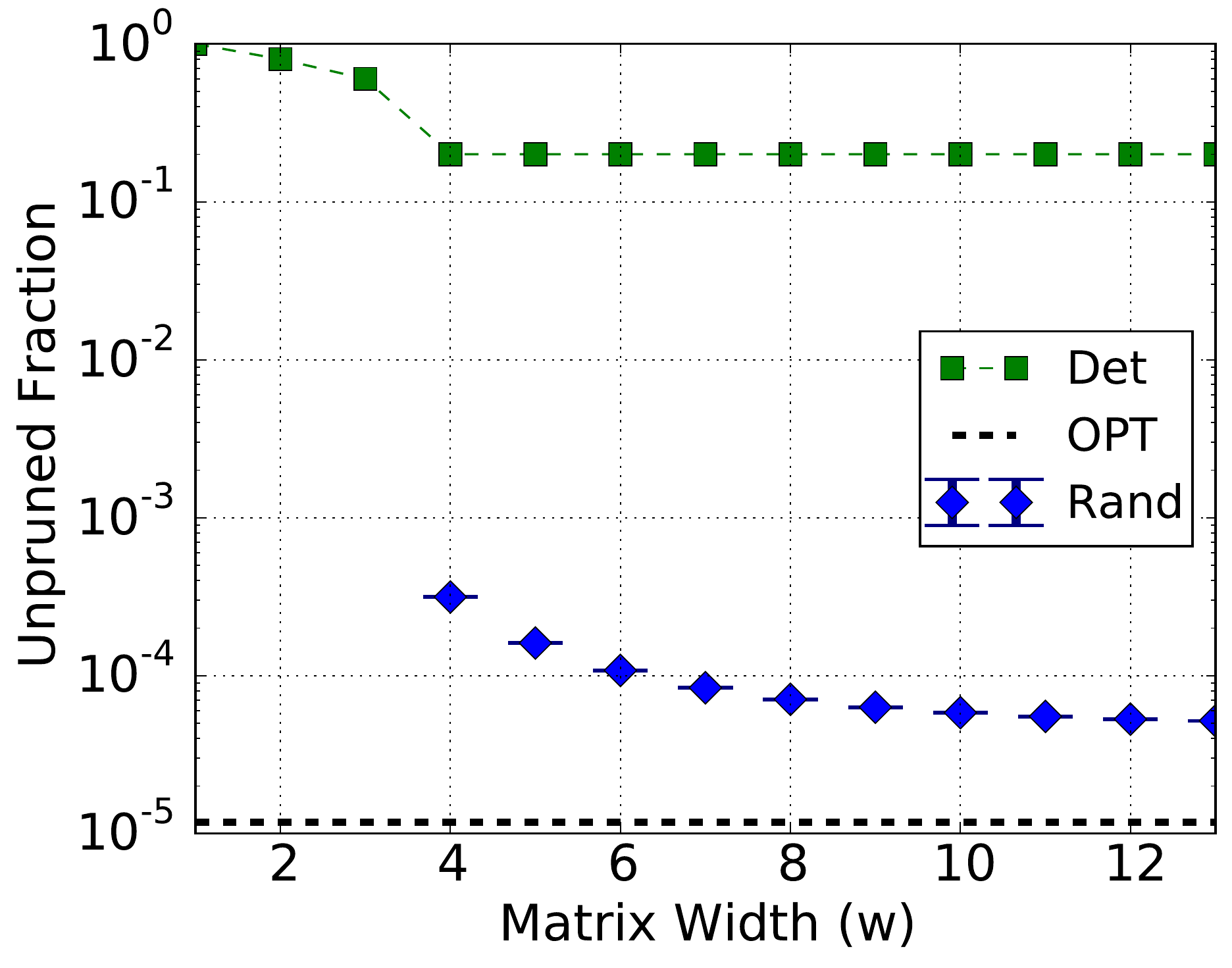}} \\
            \vspace{-4mm}
        \subfloat[\groupby]{\label{fig:groupby}\includegraphics[width =0.32\textwidth]
            {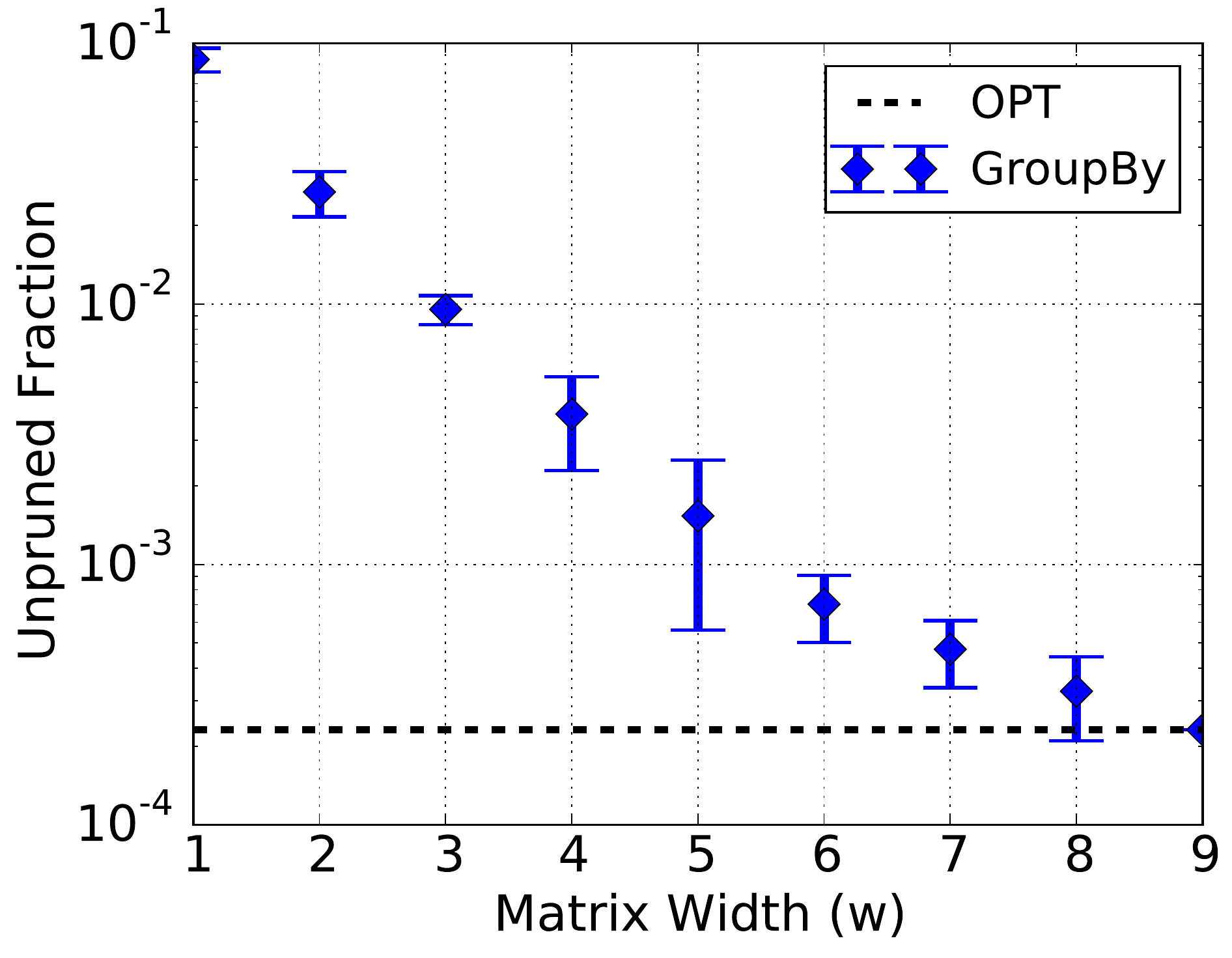}} 
        \subfloat[\join]{\label{fig:join}\includegraphics[width =0.32\textwidth]
            {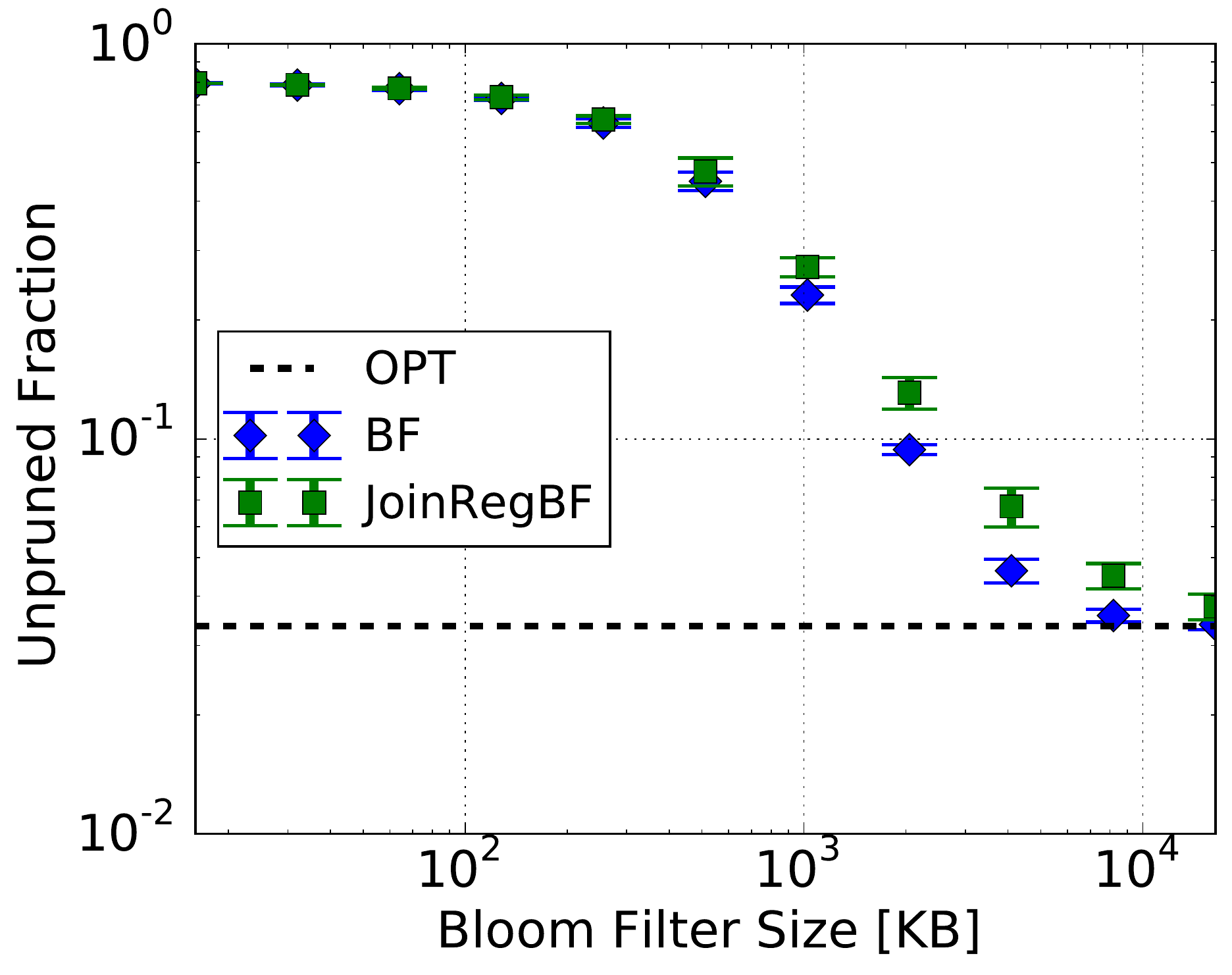}} 
        \subfloat[\having ($3$ Count Min rows)]{\label{fig:having}\includegraphics[width =0.32\textwidth]
            {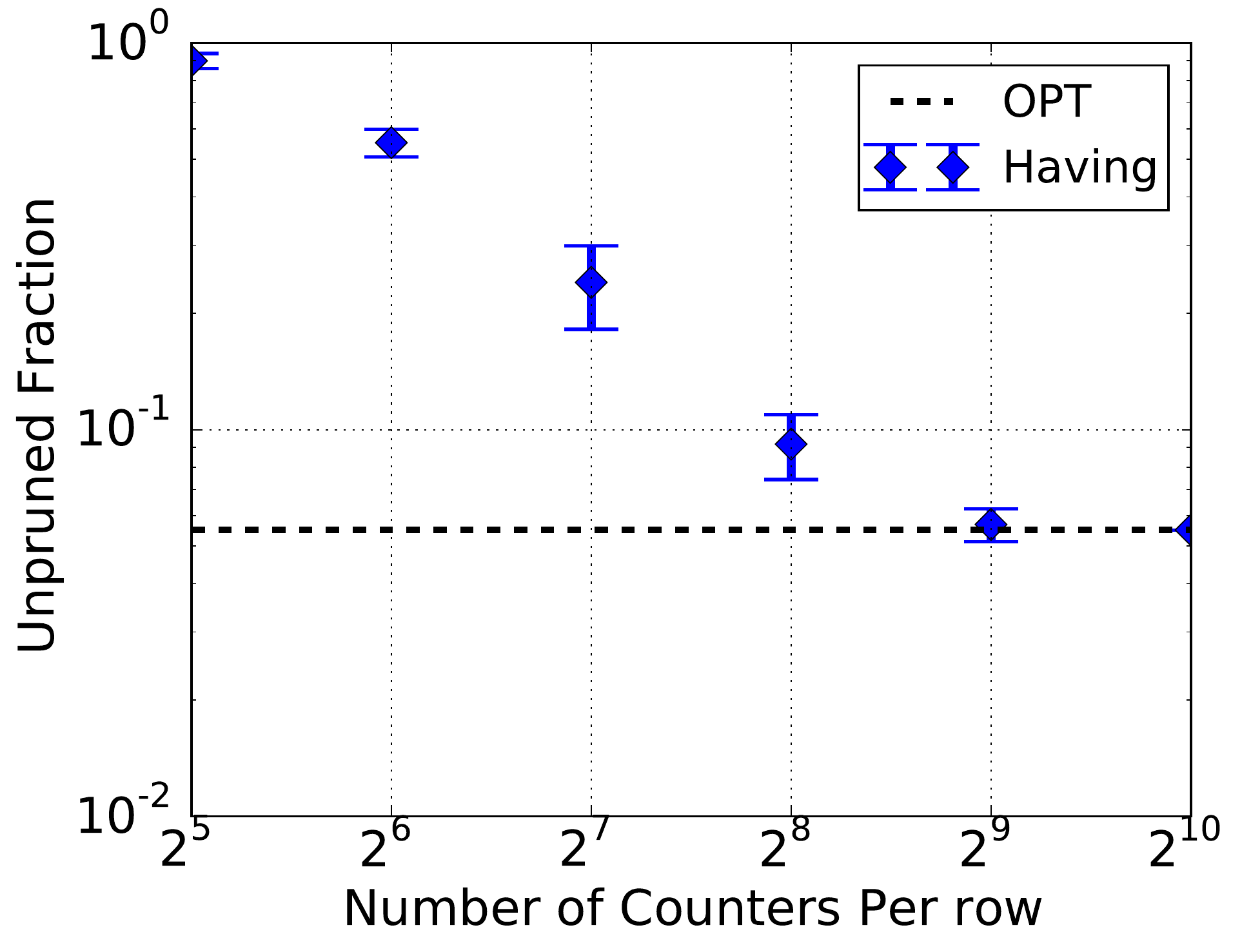}} 
    \vspace*{-2mm}
    \caption{\label{fig:Pruning} The pruning performance of our algorithms for a given resource setting. Notice that the $y$-axis is logarithmic; e.g., $10^{-3}$ means that $99.9\%$ of the \inputRows{} are pruned.}
    \vspace*{-2mm}
\end{figure*}
\else
\begin{figure*}[t]
        \subfloat[\distinct ($w=2$)]{\label{fig:distinct}\includegraphics[width =0.3\textwidth]
            {Figures/Simulations/Distinct_w=2.png}} 
        \subfloat[\skyline]{\label{fig:skyline}\includegraphics[width =0.3\textwidth]
            {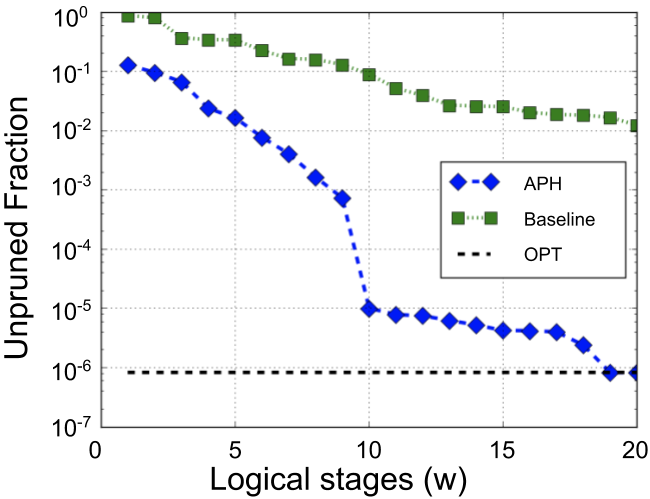}} 
        \subfloat[\topn ($d=4096$)]{\label{fig:topn}\includegraphics[width =0.3\textwidth]
            {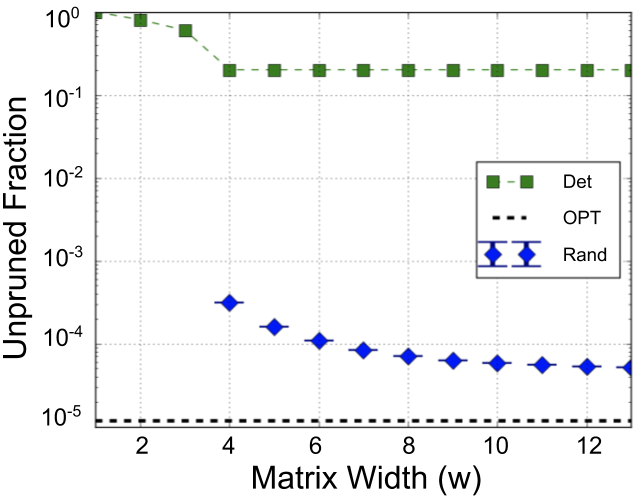}} \\
                    \subfloat[\distinct ($w=2$)]{\label{fig:distinctOverTime}\includegraphics[width =0.3\textwidth]
            {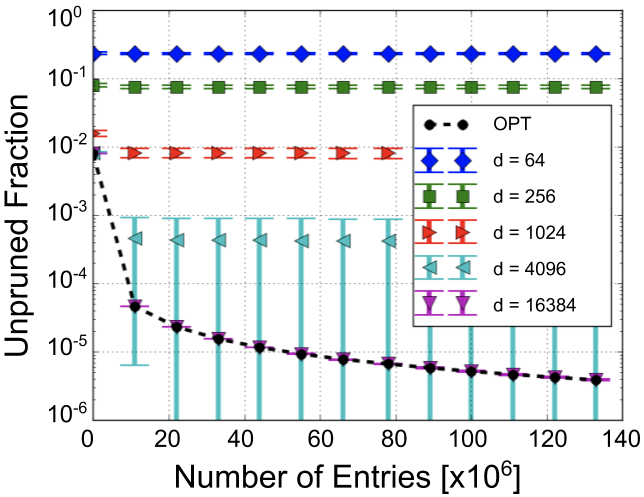}} 
        \subfloat[\skyline (APH Heuristic)]{\label{fig:skylineOverTime}\includegraphics[width =0.3\textwidth]
            {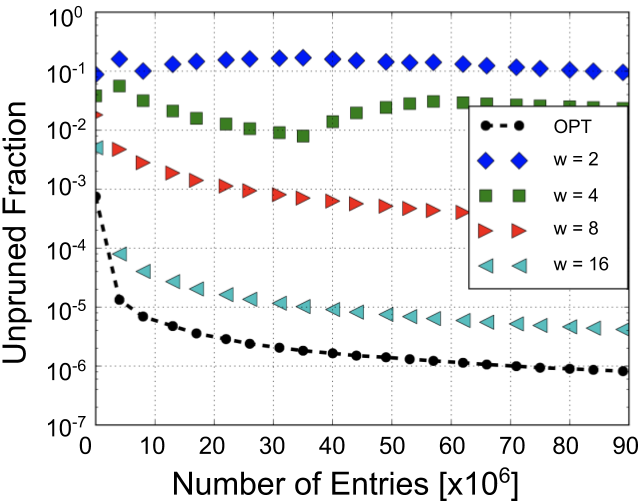}} 
        \subfloat[\topn]{\label{fig:topnOverTime}\includegraphics[width =0.3\textwidth]
            {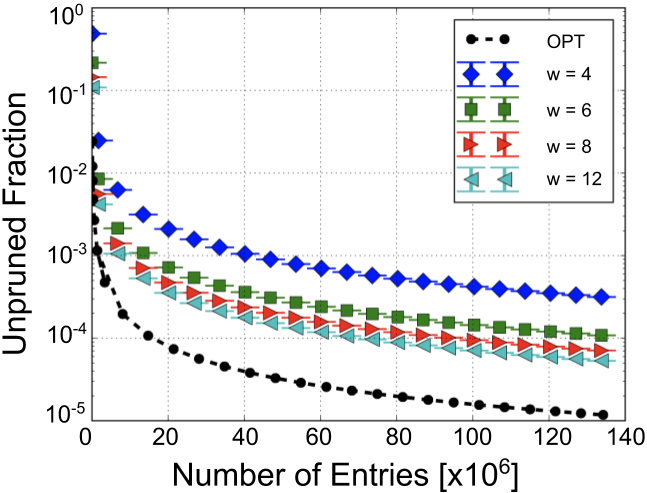}}
    \caption{\label{fig:Pruning} The pruning performance of our algorithms for a given resource setting ((a)-(c)) and vs. the data scale ((d)-(f)). Notice that the $y$-axis is logarithmic; for example, $10^{-3}$ means that $99.9\%$ of the \inputRows{} are pruned.}
\end{figure*}
\fi

\subsection{Pruning Rate Simulations}
We use simulations to study the pruning rates under various algorithm settings and switch constraints.
The pruning rates dictate the number of entries that reach the master and therefore impact the completion time. 

To understand how pruning rate affects completion time, we measure the time the master needs to complete the execution once it gets all \inputRows{}. Figure~\ref{fig:MasterTime} shows that the time it takes the master to complete the query significantly increases with more unpruned entries. The increase is super-linear in the unpruned rate since the master can handle each arriving \inputRow{} immediately when almost all entries are pruned. In contrast, when the pruning rate is low, the \inputRows{} buffer up at the master, causing an increase in the completion time.

The desired pruning rate depends on the complexity of a query's software algorithm. For example, \topn is implemented on the master using an $N$-sized heap and processes millions of entries per second. In contrast, \skyline{} is computationally expensive and thus we should prune more entries to avoid having the master become a bottleneck.
\ifdefined\shortversion
We only show three queries here. Experiments for remaining queries appear in the full version~\cite{fullVersion}.
\fi

\noindent\textbf{Pruning Rate vs. Resources Tradeoff}
We evaluate the pruning rate that \NAME achieves for given hardware constraints.
In all figures, OPT depicts a hypothetical stream algorithm with no resource constraints. For example, in a \topn it shows the fraction of \inputRows{} that were among the $N$ largest \inputRows{} from the beginning of the stream. Therefore, OPT is an upper bound on the pruning rate of \emph{any} switch algorithm. The results are depicted in Figure~\ref{fig:distinct}-\ref{fig:topn}. We ran each randomized algorithm five times and used two-tailed Student t-test to determine the 95\% confidence intervals. We configured the randomized algorithms to $\ge$99.99\% success probability.
\ifdefined\fullversion
\newcommand{\overTimeWidth}{0.325}
\begin{figure*}[t]
        \subfloat[\distinct ($w=2$)]{\label{fig:distinctOverTime}\includegraphics[width =\overTimeWidth\textwidth]
            {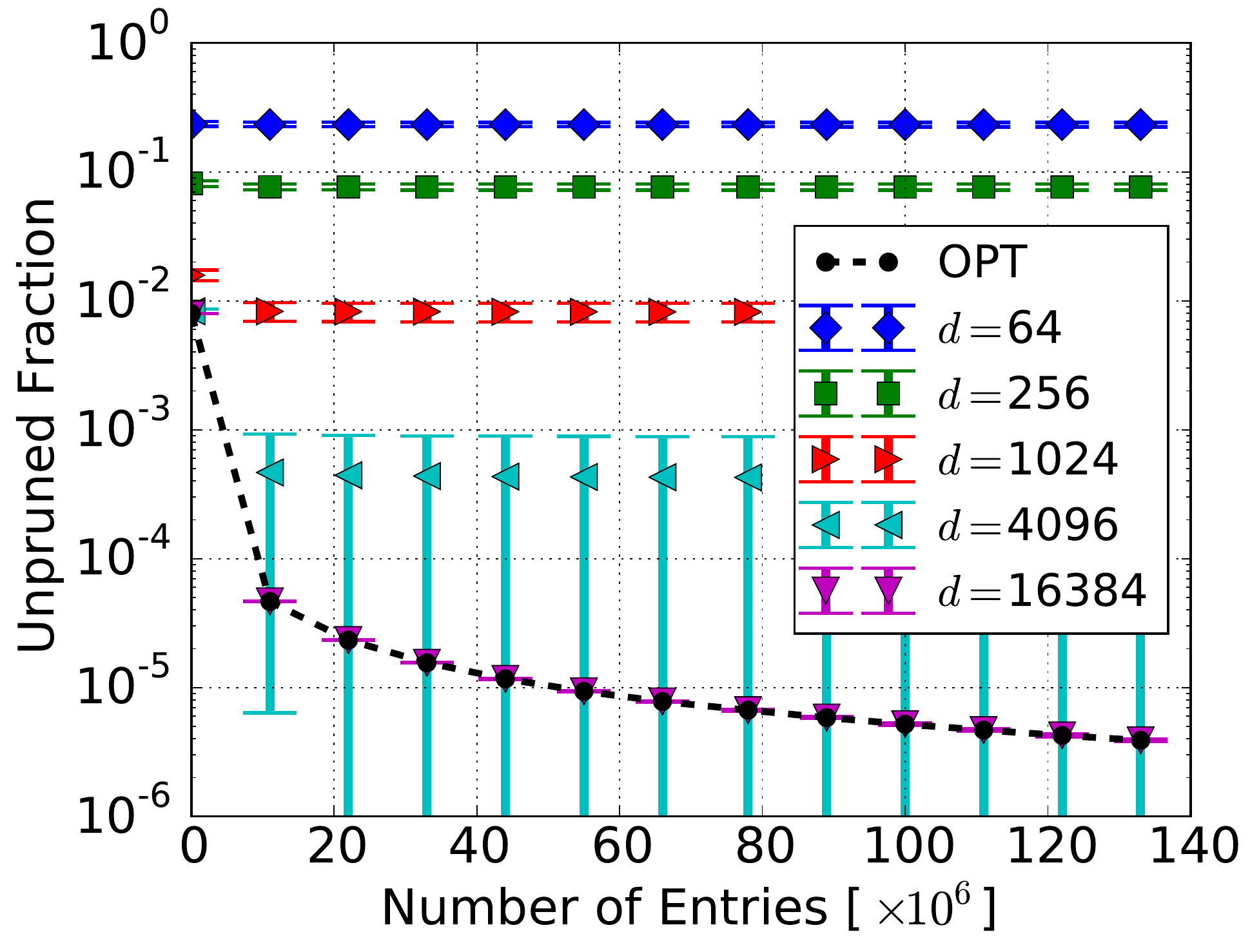}} 
            \hspace{-2mm}
        \subfloat[\skyline (APH Heuristic)]{\label{fig:skylineOverTime}\includegraphics[width =\overTimeWidth\textwidth]
            {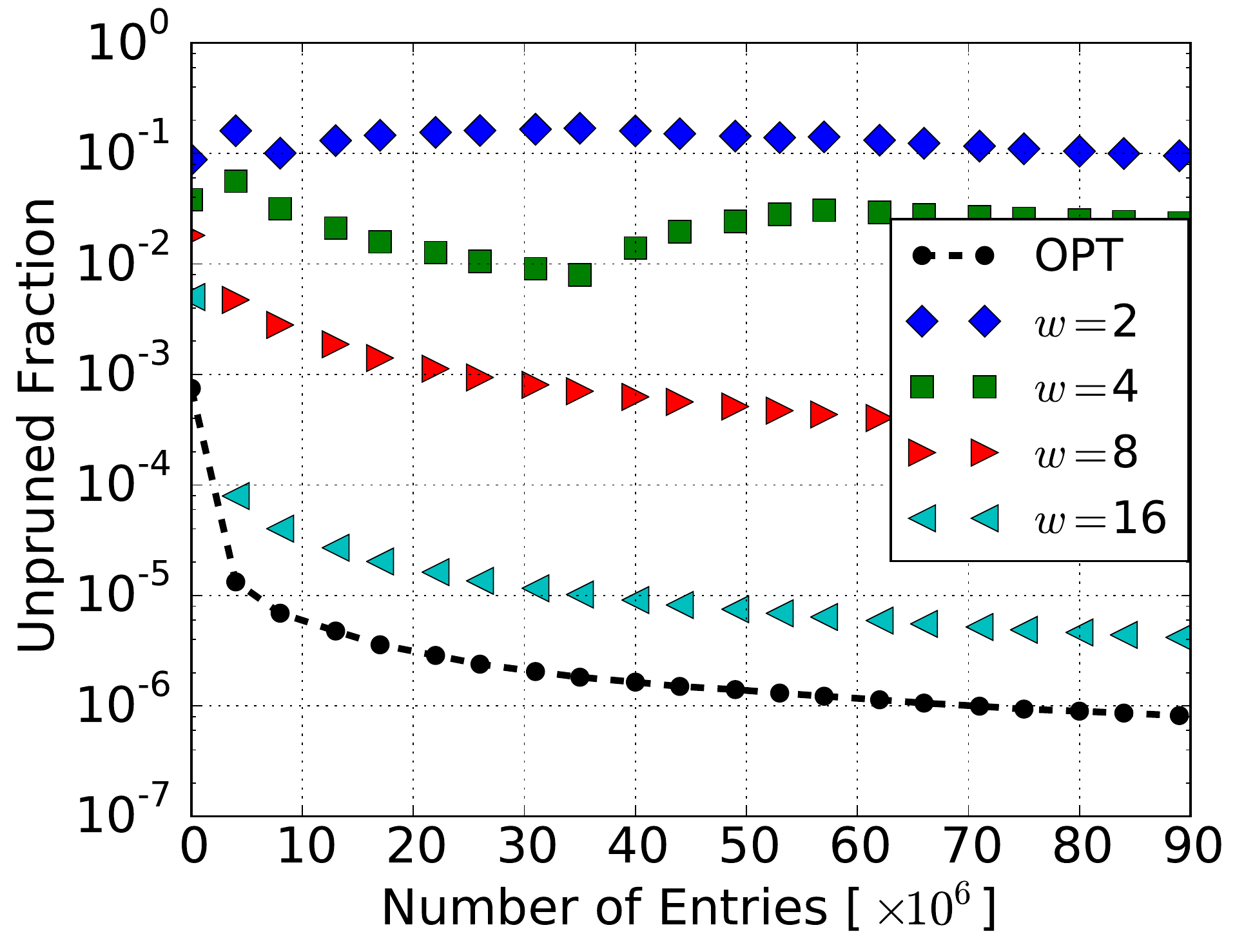}} 
            \hspace{-2mm}
        \subfloat[\topn]{\label{fig:topnOverTime}\includegraphics[width =\overTimeWidth\textwidth]
            {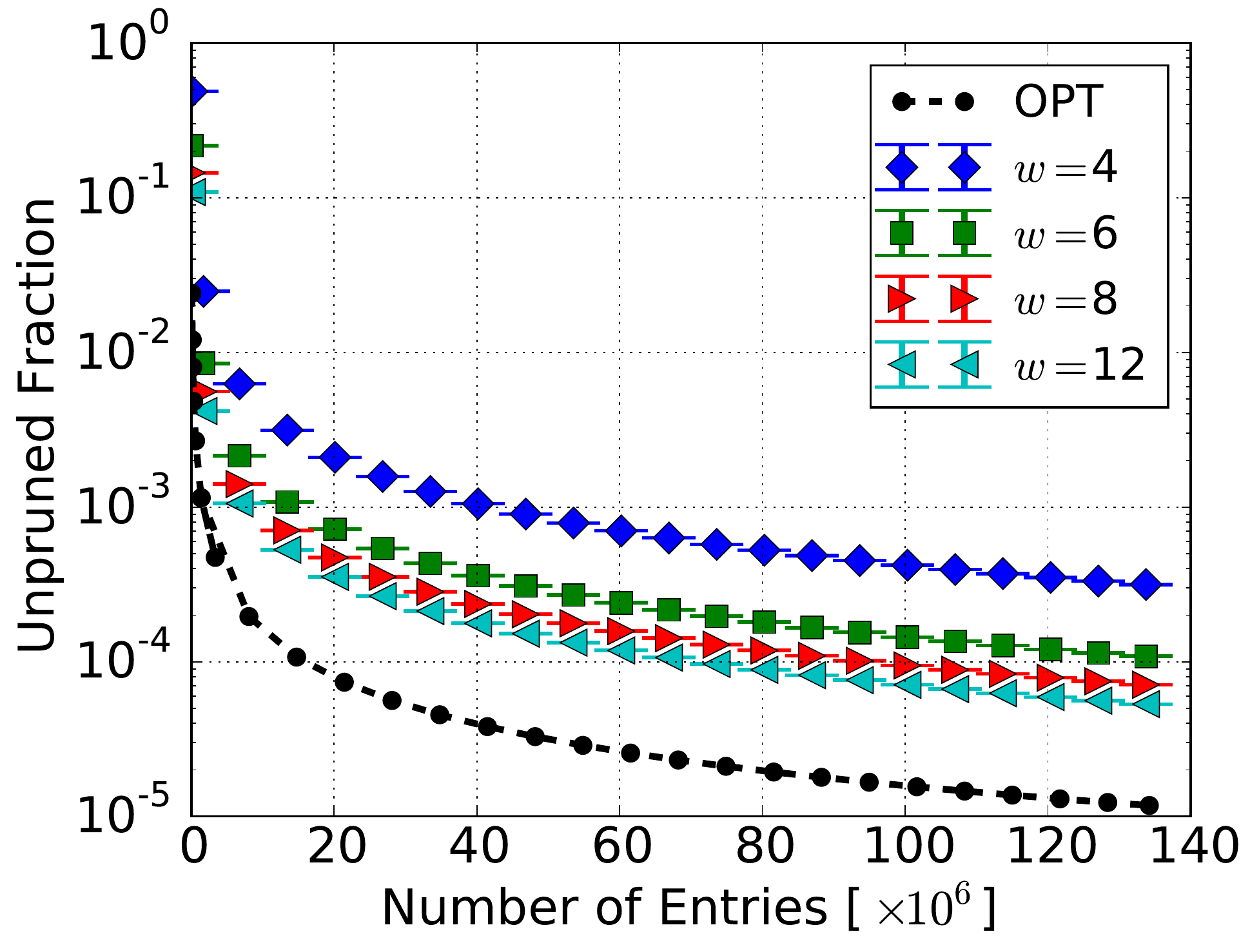}} \\
            \vspace{-4mm}
        \subfloat[\groupby]{\label{fig:groupbyOverTime}\includegraphics[width =\overTimeWidth\textwidth]
            {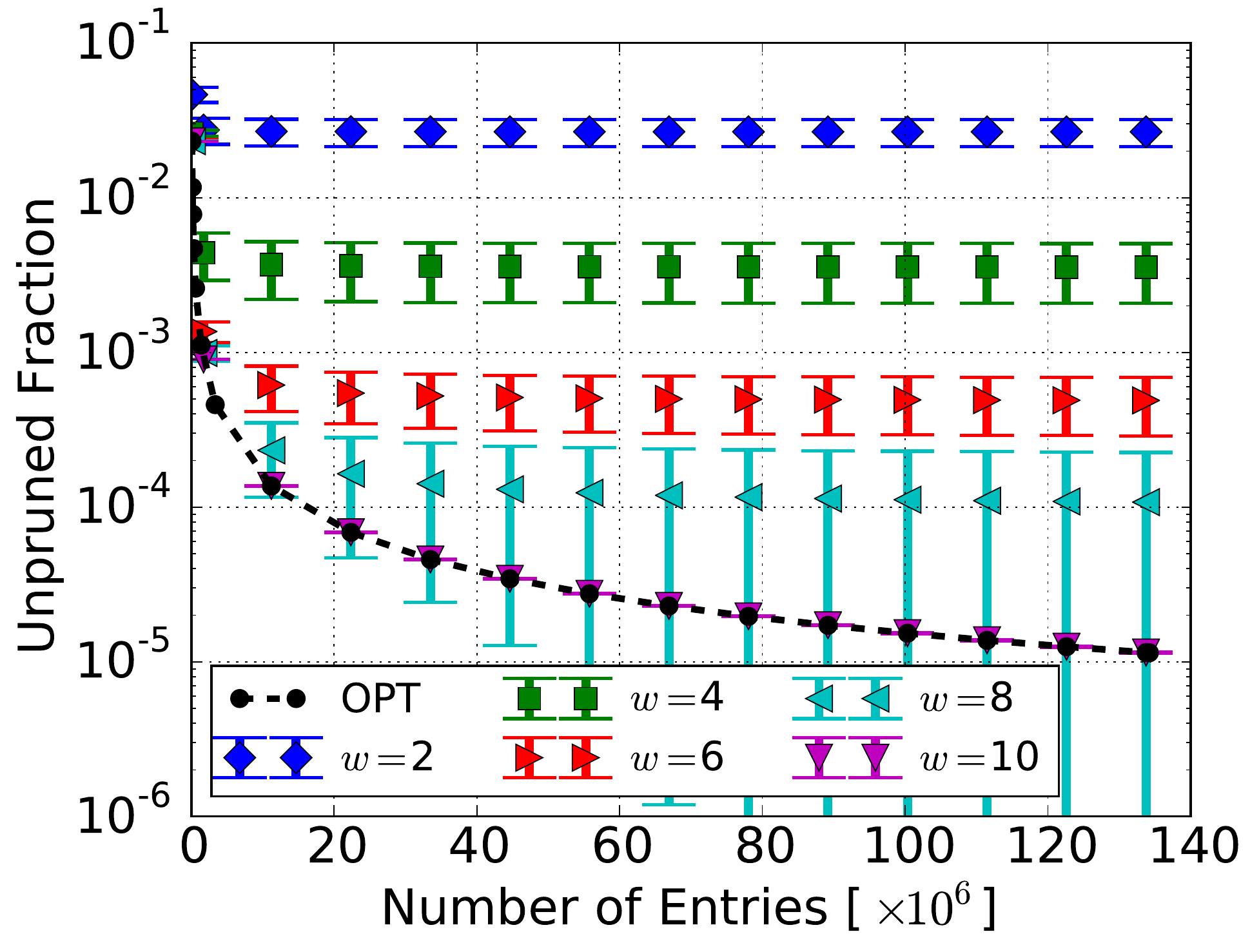}} 
            \hspace{-2mm}
        \subfloat[\join]{\label{fig:joinOverTime}\includegraphics[width =\overTimeWidth\textwidth]
                {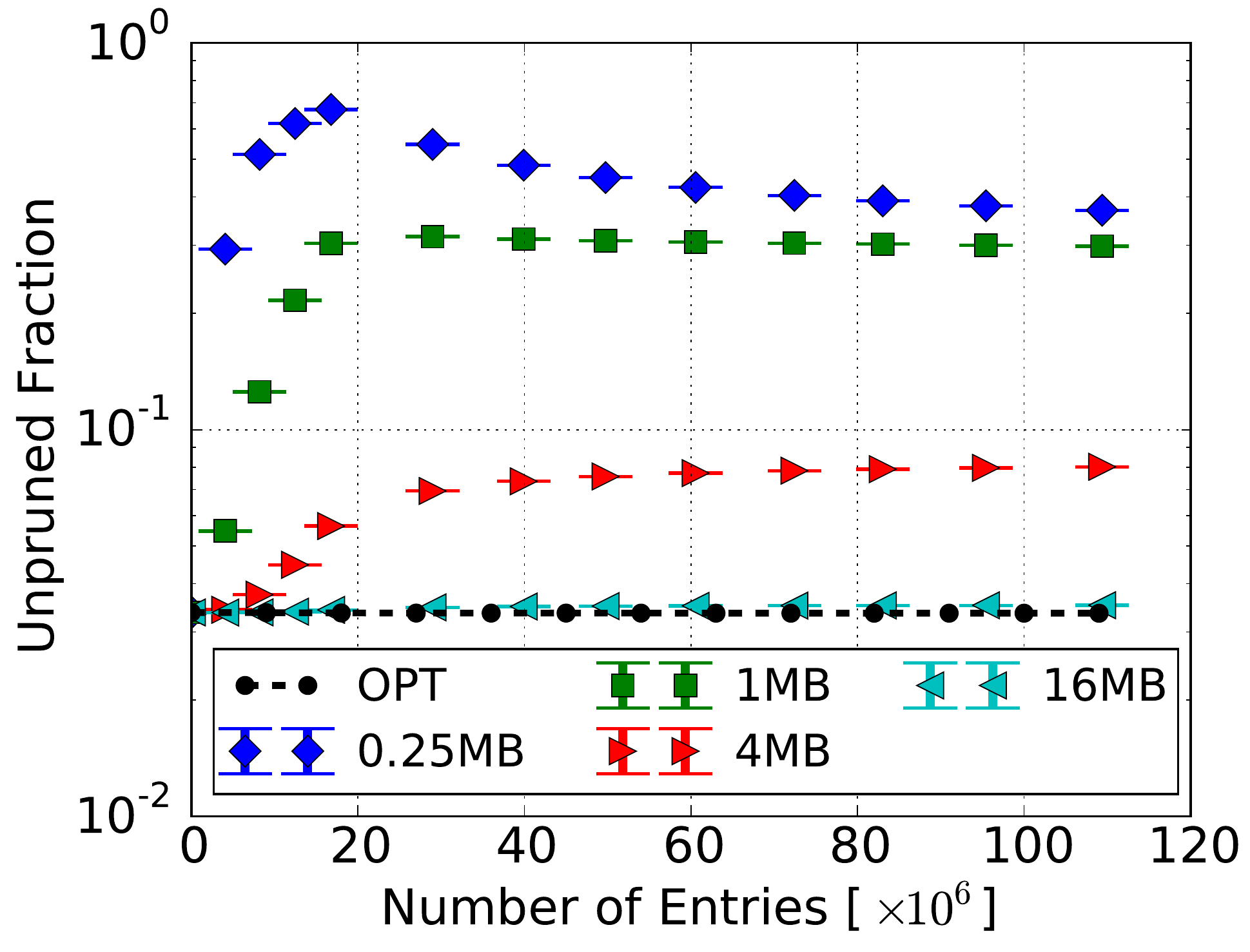}
            } 
            \hspace{-2mm}
        \subfloat[\having]{\label{fig:havingOverTime}\includegraphics[width =\overTimeWidth\textwidth]
            {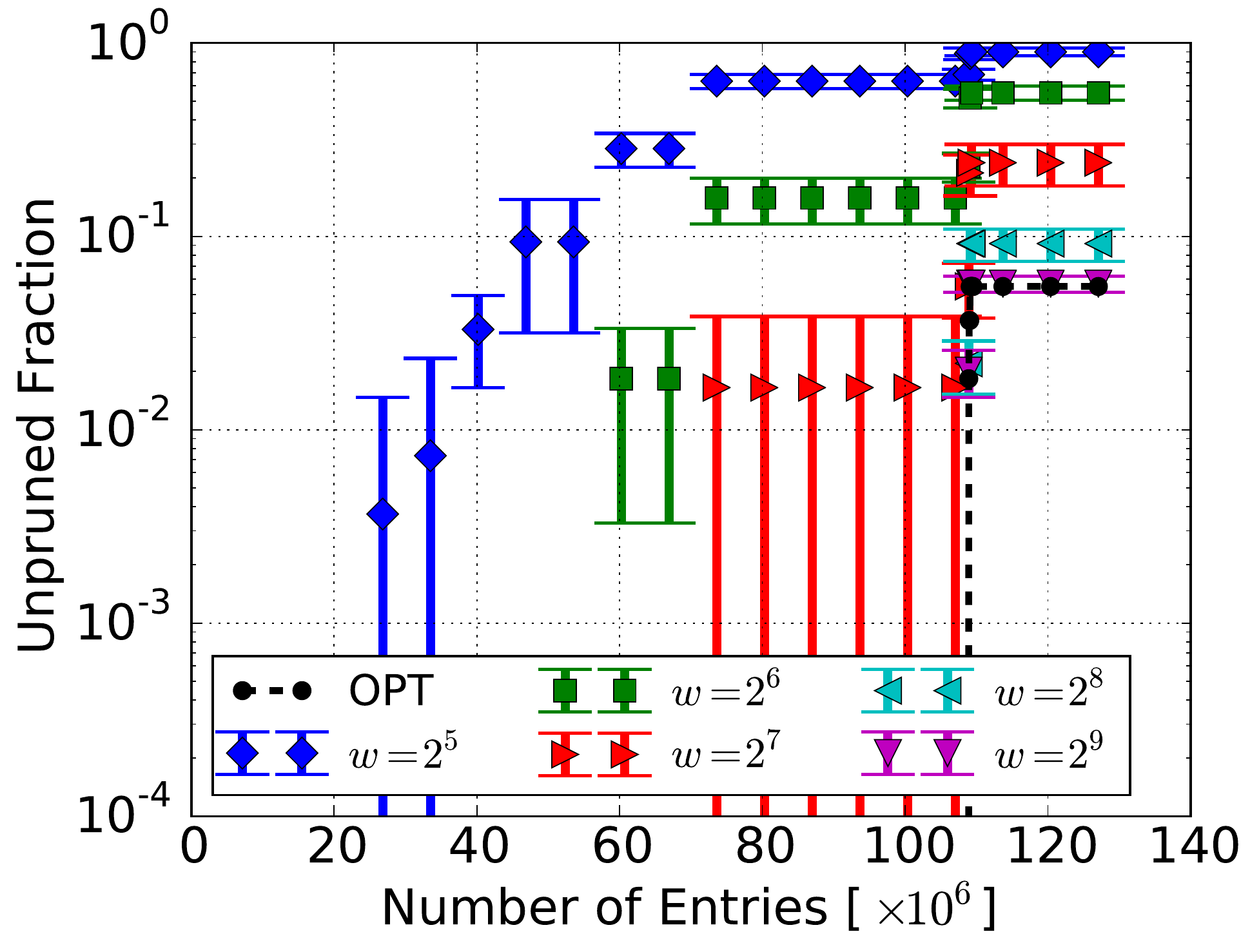}
            } 
    \caption{\label{fig:overTime}The pruning performance of our algorithms for a input size, compared to an ideal streaming algorithm with no resource constraints.}
\end{figure*}
\fi
In~\ref{fig:distinct} we see that using $w=2$ and $d=4096$ \NAME{} can prune all non-distinct entries; with smaller $d$ or the FIFO policy the pruning rate is slightly lower but \NAME{} still prunes over 99\% of the entries using just a few KB of space. In~\ref{fig:skyline} we see \skyline results; as expected, for the same number of points, APH outperforms the SUM heuristic and prunes all non-skyline points with $w=20$. Both APH and SUM prune over 99\% of the entries with $w\le 7$ while Baseline, which stands for an algorithm that store $w$ arbitrary points for pruning, requires $w=20$ for 99\% pruning. Both heuristics allow the switch to "learn" a good set of points to use for the pruning. ~\ref{fig:topn} shows \topn and illustrates the power of the randomized approach. While the deterministic algorithm can run with fewer stages and ensure correctness, if we allow just 0.01\% chance of failure, we can significantly increase the pruning rate. Here, the randomized algorithm reaches about 99.995\% pruning, leaving about $5$ times the optimal number of packets. The strict requirement for high-success probability forces the algorithm to only prune entries which are almost certainly not among the top N.
\ifdefined\fullversion
\ref{fig:groupby} shows the results for \groupby. With as few as $9$ stages \NAME{} discards all unnecessary entries \mbox{while allowing 99\% pruning with just $3$.}
Our \join performance (\ref{fig:join}) shows that \NAME{} requires at least 1MB of space to reach a good pruning rate. The performance of the Bloom Filter and the Register Bloom Filter are quite close performance wise and both reach near-optimal pruning with 16MB of space. For \having{}, as we see in~\ref{fig:having}, \NAME{} gets perfect pruning using $1024$ counters in each \mbox{of the three rows.}
\fi

\noindent\textbf{Pruning Rate vs. the Data Scale:}
The pruning rate of different algorithms behaves differently when the data scale grows. Here, each data point refers to the first entries in the relevant data set.
\ifdefined\fullversion
Specifically, as we show in Figure~\ref{fig:Pruning}, the pruning rate of several algorithms goes lower while others improve as the scale grows. 
 \ref{fig:distinctOverTime}, \ref{fig:skylineOverTime}, \ref{fig:topn}, and \ref{fig:groupbyOverTime} show that for \distinct, \skyline, \topn, and \groupby{} \NAME{} achieves better pruning rate for larger data. For \distinct and \groupby it is because we cannot prune the first occurrence of an output key, but once our data structure has these reflected it gets better pruning in the future. In \skyline and \topn, a smaller fraction of the input entries are needed for the output as the data scale grows, allowing the algorithms to prune more entries.
In contrast, the algorithms for \join and \having (\ref{fig:joinOverTime} and~\ref{fig:havingOverTime}) have better pruning rates for \emph{smaller} data sets. In \join, the algorithm experiences more false positives as the data keep on coming and therefore prunes a smaller fraction of the entries. The \having query is conceptually different; as it asks for the codes for languages whose sum-of-ad-revenue is larger than \$1M, the output is \emph{empty} if the data is too small. The one-sided error of the Count Min sketch that we use guarantees that we do not miss any of the correct output keys but the number of false positives increases as the data grows. Nevertheless, with as few as $512$ counters for each of the three rows, \NAME{} gets near-perfect pruning throughout the evaluation.
\else
Figures \ref{fig:distinctOverTime} - \ref{fig:topnOverTime} show how the pruning rate varies as the scale grows. For the shown queries, \NAME{} achieves better pruning rate for larger data. For \distinct it is because we cannot prune the first occurrence of an output key, but once our data structure has these reflected it gets better pruning. In \skyline and \topn, a smaller fraction of input entries are needed for the output as the data scale grows, allowing the pruning of more entries.
In contrast, the algorithms for \join and \having have better pruning rates for \emph{smaller} data sets, as shown in the full version~\cite{fullVersion}. In \join, the algorithm experiences more false positives as the data keep on coming and therefore prunes a smaller fraction of the entries. The \having query is conceptually different; as it asks for the codes for languages whose sum-of-ad-revenue is larger than \$1M, the output is \emph{empty} if the data is too small. The one-sided error of the Count Min sketch that we use guarantees that we do not miss any of the correct output keys but the number of false positives increases as the data grows. Nevertheless, with as few as $512$ counters for each of the three rows, \NAME{} gets near-perfect pruning throughout the evaluation.
\fi
\section{Extensions}\label{sec:extensions}
\parab{Multiple switches:}
We have considered a single programmable switch in the path between the workers and the master. However, having multiple switches boosts our performance further. For example, we can use a ``master switch'' to partition the data and offload each partition to a different switch. Each switch can perform local pruning of its partition and return it to the master switch which prunes the data further. This increases the hardware resources at our disposal and allows superior pruning results.

\parab{DAG of workers:}
Our paper focuses on the case where there is one master and multiple workers. However, in large scale deployments or complex workloads, query planning may result in a directed acyclic graph (DAG) of workers, each takes several inputs, runs a task, and outputs to a worker on the next level. In such cases, we can run \NAME{} at each edge in which data is sent between workers. To distinguish between edges, each has a dedicated port number and a set of resources (stages, ALUs, etc.) allocated to it. To that end, we use the same packing algorithm described in~\cref{sec:multiqueries}.

\parab{Packing multiple entries per packet:}
\NAME{} spends a significant portion of its query processing time on transmitting the entries from the workers. This is due to two factors; first, it does not run tasks on the workers that filters many of the entries; second, it only packs one entry in each packet. 
While the switch cannot process a very large number of entries per packet on the switch due to limited ALUs, we can still pack several (e.g. four) entries in a packet thereby significantly reducing this delay. P4 switches allow popping header fields~\cite{P4Spec} and thereby support pruning of some of the entries in a packet. The limit on the number of entries in each packet depends on the number of ALUs per stage (all our algorithms use at least one ALU per entry per stage), the number of stages (we can split logical stage to several if the pipeline is long enough).
Our \distinct, \topn, and \groupby algorithms support multiple entries per packet while maintaining correctness: if several entries are mapped to the same matrix row, we can avoid processing them while \mbox{not pruning the entries.}
\section{Related Work}\label{sec:related}\
This work has not been published elsewhere except for a 2-page poster at SIGCOMM~\cite{sigcommPoster}. The poster discusses simple filtering, \distinct, and \topn. This work significantly advances the poster by providing pruning algorithms for $4$ additional queries, an evaluation on two popular benchmarks, and a comparison with NetAccel~\cite{lerner2019case}. This work also discusses probabilistic pruning, optimizing multiple queries, using multiple switches, and a reliability protocol.

\parab{Hardware-based query accelerators:} {\NAME} follows a trend of accelerating database computations by offloading computation to hardware. Industrial systems ~\cite{Netezza, Exadata} offload parts of the computation to the storage engine. Academic works suggest offloading to FPGAs~\cite{Cipherbase,dennl2012fly,sukhwani2012database,woods2014ibex}, SSDs~\cite{do2013query}, and GPUs~\cite{paul2016gpl,Govindaraju:2004:FCD:1007568.1007594,Sun:2003:HAS:872757.872813}. These either consider offloading the entire operation to hardware~\cite{lerner2019case, woods2014ibex}, or doing a per-partition \emph{exact} aggregation/filtering before transmitting the data for aggregation~\cite{do2013query}. However, exact query computation on hardware is challenging and these only support basic operations (e.g., filtering~\cite{Netezza, do2013query}) or primitives  (e.g., partitioning~\cite{FPGAPartitioning} or aggregation~\cite{dennl2013acceleration, woods2014ibex}).

 \NAME uses programmable switches which are either cheaper or have better performance than alternative hardware such as FPGAs. Compared to FPGAs, switches handle two orders of magnitude more throughput per Watt~\cite{Tokusashi:2019:CIC:3302424.3303979} and ten times more throughput per dollar~\cite{switchcost}~\cite{Tokusashi:2019:CIC:3302424.3303979}. GPUs consume 2-3x more energy than FPGAs for equivalent data processing workloads~\cite{Owaida:2019:LLD:3357377.3365457} and double the cost of FPGA with similar characteristics~\cite{Owaida:2019:LLD:3357377.3365457}. A summary of the attributes of the different alternatives appears in Table~\ref{tab:tofino}. Switches are also readily available in the networks at no additional cost. We offload {\em partial} functions on switches using the pruning abstraction and support a variety of database queries.
\begin{table}
    \centering
    \resizebox{1.025\columnwidth}{!}{
    \begin{tabular}{|l|l|l|l|l|l|l|}\hline
      System   & Server & GPU~\cite{gpubw} & FPGA~\cite{netfpga} & SmartNIC~\cite{mellanox} & 
      Tofino V2~\cite{tofinov2}\\\hline
      Throughput & 10-100Gbps & 40-120Gbps & 10-100Gbps & 10-100Gbps & 
      12.8 Tbps\\
      Latency & 10-100 $\mu$s & 8-25 $\mu$s & 10 $\mu$s & 5-10$\mu$s 
      &   $<$ 1$\mu$s \\\hline
    \end{tabular}
    }
    \caption{Performance comparison of hardware choices. 
    }
    \label{tab:tofino}
    \vspace*{-7mm}
\end{table}
Sometimes, FPGAs and GPUs also incur extra data transmission overhead. For example, GPU's separate memory system also introduces significant performance overhead with extra computation and memory demand~\cite{woods2014ibex}. When an FPGA is attached to the PCIe bus (~\cite{Cipherbase, Netezza}), we have to copy the data to and from the FPGA explicitly~\cite{woods2014ibex}. One work has used FPGAs as an in-datapath accelerator to avoid this transfer~\cite{woods2014ibex}.

Moreover, switches can see the aggregated traffic across workers and the master, and thus allow optimizations across data partitions. In contrast, FPGAs are typically connected to individual workers due to bandwidth constraints ~\cite{Netezza} and can only optimize the query for each partition.
That said, \NAME{} complements these works as switches can be used with FPGA and GPU-based solutions for additional performance.

\parab{Offloading to programmable switches:}
Several works use programmable switches for offloading different functionality that was handled in software~\cite{NetCache, DistCache, NetChain, NetPaxos, SilkRoad, Jepsen:2018:LFL:3185467.3185494, Jepsen:2018:PSP:3286062.3286092, Snappy, PRECISION, harrison2018network}. \NAME
offloads database queries, which brings new challenges to fit the constrained switch programming model because database queries often provide a large amount of data and require diverse computations across many entries. One opportunity in databases is that the master can complete the query from the pruned data set.

In the network telemetry context, researchers proposed Sonata, a general monitoring abstraction that allows scripting for analytics and security applications~\cite{sonata}. Sonata supports filtering, map and a constrained version of distinct in the data plane but relies on a software stream processor for other operations (e.g., \groupby, \topn, \skyline). Conceptually, as Sonata offloads only operations that can be fully computed in the data plane, its scope is limited. Sparser~\cite{palkar2018filter} accelerates text-based filtering using SIMD.

Recently, NetAccel~\cite{lerner2019case} suggested using programmable switches for query acceleration. We discuss and evaluate the differences between \NAME{} and NetAccel in~\cref{sec:netaccel}. Jumpgate ~\cite{mustard2019jumpgate} also suggests accelerating Spark using switches. It uses a method similar to NetAccel. However, while \NAME{} and NetAccel are deployed in between the worker and master server, Jumpgate stands between the storage engine and compute nodes. Jumpgate does not include an implementation, is specific to filtering and partial aggregation, and cannot cope with packet loss.
\section{Conclusion}
We present \NAME, a new query processing system that significantly reduces query completion time compared to the current state-of-the-art for a variety of query types. 
\NAME accelerates queries by leveraging programmable switches while using a pruning abstraction to fit in-switch constraints without affecting query results. 
 
\section{Acknowledgements}
We thank the anonymous reviewers for their valuable feedback. We thank Mohammad Alizadeh and Geeticka Chauhan for their help and guidance in the early stages of this project, and Andrew Huang for helping us validate Spark SQL queries. This work is supported by the National Science Foundation under grant CNS-1829349 and the Zuckerman foundation.

\bibliographystyle{abbrv}
\bibliography{bibliography}
\ifdefined\fullversion
\clearpage
\appendix

\section{Algorithm Summary Table}\label{app:algTable}
\subsection{Algorithm Implementation}
\begin{table}[H]
    \centering
    \resizebox{1.025\columnwidth}{!}{
    \begin{tabular}{|l|c|c|p{3.33cm}|
    }\hline
      \textbf{Algorithm} & \textbf{Guarantee} & \textbf{Parameters} & \textbf{Meaning}
      \\\hline\hline
      \distinct{}&Rand&$(w,d)$&A $d\times w$ matrix used as a $w$-way cache 
      
      \\\hline
      \skyline{}&Det&$w$&Number of points stored on the switch
      \\\hline
      \multirow{2}{*}{\topn{}}&Det&$w$&Number of counters stored on the switch
      \\\hhline{~---}
      &Rand&$(w,d)$&A $d\times w$ matrix where each row uses a rolling minimum.
     \\\hline
      \groupby{}&Det&$(w,d)$&$d\times w$ matrix with one hash per row.
      \\\hline
      \join{}&Det&$(M,H)$& $M$ filter bits,\qquad{}\qquad{} $H$ hash functions
      \\\hline
      \having{}&Det&$(w,d)$&Count Min Sketch with $d$ rows and $w$~columns\\
       \hline
    \end{tabular}
    }
    \caption{Summary of our algorithms}
    \label{tab:algSummary}
\end{table}

\subsection{Implementation Resource Estimates}

\subsubsection{Pruning overhead}
All the discussed queries require one extra Match-Action table that checks a single-bit metadata entry, $prune$, and drops the packet if it is true. The total metadata requirements for all queries combined is comparable to that of an IPv4 header and no individual query in our implementation took more than $\sim$ 255 bits of metadata. Some queries have metadata that scales with the SRAM we allocate to it, we will discuss these in the individual resource estimates for those queries.  

\subsubsection{Filtering}

Filtering a single condition requires just 1 ALU. If we want to allow the control plane to reconfigure the filtering constant $c$ in a query such as 
\begin{lstlisting}[backgroundcolor = \color{lightgray!50},style=SQL, breaklines=true, aboveskip=0pt, belowskip=0pt]
SELECT * WHERE x > c
\end{lstlisting}
 we need a register to store the value of $c$. Otherwise, if we do not need to change the constant at runtime, filtering takes no additional SRAM .

\subsubsection{Distinct}
To implement a single LRU cache that can contain $w$ entries and can store upto $64$-bit values, we allocate a $64$-bit register on $w$ stages. To implement $d$ LRU caches (i.e., a multi-row LRU cache), we can increase the number of $64$-bit registers to $d$ together with a hash to select the right register index to map a value to. This results in the total SRAM requirement for a $d$-row $w$-column multi-row LRU cache being $d \times w \times 64b$.
The rolling replacement policy we use to ensure LRU eviction (section ~\ref{sec:querypruningalgorithms}) requires an ALU per stage. Multiple rows do not require additional ALUs as the same ALU can be provided a different register index depending on the row. Therefore, ALUs equal to the number of stages i.e., $w$ are used.

\subsubsection{Group By}
To implement Group By, we need to allocate $d$ (64 bit) registers per stage to store values (where $d$ is the number of $64$-bit register indexes available) along with an ALU per stage to compare with a currently received packet's value. Hence a 1 stage Group By requires at least $d \times 64b$ SRAM and 1 ALU. A $w$ stage Group By requires $w \times d \times 64b$ SRAM and $w$ ALUs.

\subsubsection{Deterministic Top N}
Every cutoff, $t_{i}$ requires a register in a separate stage. We require an ALU at every stage to count the number of packets seen above that entry and confirm we saw at least $N$ before pruning on that stage's $t_{i}$. We need one extra stage and ALU to conduct the $min$ computation required to set $t_{0}$. This results in a total of $w + 1$ ALUs and $w + 1$ stages. Since we use one register per stage and we use $64$ bit registers, the SRAM used is $(w + 1) \times 64b$.

\subsubsection{Randomized Top N}
This query uses a similar hash-based register indexing mechanism as \distinct{} hence we $d$ registers per stage. If the registers are of width $64$ bits and we have $w$ stages, this translates to a total SRAM usage of $d \times w \times 64b$ similar to \distinct{}. To conduct the comparison and replacement operation neccessary for implementing a rolling minimum, we need to use an ALU at each stage, resulting in a total of $w$ ALUs used.
\section{The Queries used for our Benchmarking}\label{app:queries}
The benchmark consists of two tables that we use -- Rankings and UserVisits. 
Ranking consists of 90M \inputRows{} with three \inputCols{}: pageURL, pageRank, avgDuration and is roughly sorted on pageRank. UserVisits has 775M \inputRows{} with nine 
\inputCols{}, including destURL, adRevenue, languageCode, and userAgent. 
We consider the following queries: 

{\scriptsize
\begin{lstlisting}[backgroundcolor = \color{lightgray!50},style=SQL, breaklines=true, aboveskip=0pt, belowskip=0pt]
(1) SELECT COUNT() FROM Rankings WHERE avgDuration < 10
\end{lstlisting}
\begin{lstlisting}[backgroundcolor = \color{lightgray!50},style=SQL, breaklines=true, aboveskip=0pt, belowskip=0pt]
(2) SELECT DISTINCT userAgent FROM UserVisits
\end{lstlisting}
\begin{lstlisting}[backgroundcolor = \color{lightgray!50},style=SQL, breaklines=true, aboveskip=0pt, belowskip=0pt]
(*\footnote{As the data is nearly sorted on pageRank, we run the query on a random permutation of the table.}*)(3) SELECT * FROM Ratings (*\bfseries SKYLINE*) (*\bfseries OF*) pageRank,avgDuration 
\end{lstlisting}
\begin{lstlisting}[backgroundcolor = \color{lightgray!50},style=SQL, breaklines=true, aboveskip=0pt, belowskip=0pt]
(4) SELECT (*\bfseries TOP*) 250 * FROM UserVisits ORDER BY adRevenue
\end{lstlisting}
\begin{lstlisting}[backgroundcolor = \color{lightgray!50},style=SQL, breaklines=true, aboveskip=0pt, belowskip=0pt]
(5) SELECT userAgent, MAX(adRevenue) FROM UserVisits GROUP BY userAgent
\end{lstlisting}
\begin{lstlisting}[backgroundcolor = \color{lightgray!50},style=SQL, breaklines=true, aboveskip=0pt, belowskip=0pt,escapechar=ß]
ß\footnote{As the data have 100\% match between the keys, we took a random 10\% subset of each table for the join.}ß(6) SELECT * FROM UserVisits JOIN Ratings ON UserVisits.destURL = Ratings.pageURL
\end{lstlisting}
\begin{lstlisting}[backgroundcolor = \color{lightgray!50},style=SQL, breaklines=true, aboveskip=0pt, belowskip=0pt]
(7) SELECT languageCode FROM UserVisits GROUP BY languageCode HAVING SUM(adRevenue) > 1000000
\end{lstlisting}
}
\section{Analysis of \NAME's \distinct Algorithm}\label{app:distinct}
Our goal is to satisfy the \emph{probabilistic} accuracy guarantee, which means that with high probability we return the correct answer without pruning any distinct \inputRow{}. To that end, we observe that it is enough to guarantee that on all \matrixRows{}, no two \inputRows{} share the same fingerprint. 

\begin{theorem}
Consider a stream with $m$ entries and our algorithm with $w$ columns and fingerprints of size $f=\ceil{\log\parentheses{w\cdot m / \delta}}$. Then with probability at least $1-\delta$ there are no false positives.
\end{theorem}
\begin{proof}
Consider an arbitrary item in the input. If the row that it is mapped to is full, according to the union bound, it has a probability of at most $w\cdot 2^{-f}$ of colliding with one of the stored items (not including itself). Taking the union bound over all the data, we have that the probability that any item had a fingerprint collision is at most \mbox{$m\cdot w\cdot 2^{-f}\le\delta$.}
\end{proof}

Sometimes, the size of the data is too large for the above analysis to provide feasible fingerprint length due to the logarithmic dependency on the number of tuples. Instead, we can analyze the fingerprint size required given that the number of \emph{distinct} elements is small while the input can be arbitrarily large.

Denoting by $L$ the fingerprint length and by $X_i$ the number of \emph{distinct} entries mapped into \matrixRow{} $i$, we have that the probability of fingerprint collision in it is bounded by $2^{-L}\cdot {X_i \choose 2}\le 2^{-L-1}\cdot  X_i^2$. Using the union bound, we have that the chance that in any \matrixRow{} there is such a collision is at most $2^{-L-1}\cdot \sum_{i=1}^d X_i^2$. Therefore, if we set the fingerprints to be of size $\logp{\frac{\sum_{i=1}^d X_i^2}{2\delta}}$ we will get that with probability $1-\delta$ the algorithm is correct. However, $\sum_{i=1}^d X_i^2$ is a random variable and not a fixed quantity. Our analysis discusses how to upper-bound $\sum_{i=1}^d X_i^2$, with probability $1-\delta/2$, in a Balls-and-Bins experiment in which $D$ balls (distinct \inputRows{}) are thrown randomly into $d$ bins (\matrixRows{}). Then, we require that with probability $1-\delta/2$ there will be no same-\matrixRow{} fingerprint collisions \emph{given our bound on $\sum_{i=1}^d X_i^2$}. Using the union bound we conclude that the algorithm will succeed with probability of at least $1-\delta$. Unlike that $\approx 2\log D$-sized fingerprints that would be required if one requires all fingerprints to be distinct, we show that $(1+o(1)) \log (D^2/d)$ bits are enough (i.e., we save $\approx \log d$ \mbox{bits on each fingerprint).}

For example, if $d=1000$ and $\delta=0.01\%$, we can support up to $500M$ distinct elements using $64$-bits fingerprints \emph{regardless of the data size}. Further, this does not depend on the value of $w$.

\begin{theorem}\label{thm:distinct}
Denote$$
\mathcal M\triangleq \begin{cases}
e\cdot D/d &\mbox{if $D>d\ln (2d/\delta)$}\\
e\cdot \ln (2d/\delta) &\mbox{if $d\cdot\ln\delta^{-1}/e \le D\le d\ln (2d/\delta)$}\\
\frac{1.3\lnp{2d/\delta}}{\lnp{\frac{d}{D\cdot e}\lnp{2d/\delta}}}&\mbox{otherwise}
\end{cases} ,
$$ where $D$ is the number of distinct items in the input.
Consider storing fingerprints of size $f=\ceil{\log(d\cdot \mathcal M^2/\delta)}$
bits. Then with probability $1-\delta$ there are no false positives and the distinct operation terminates successfully.
\end{theorem}


\begin{theorem}\label{thm:distinct}
Denote$$
\mathcal M\triangleq \begin{cases}
e\cdot D/d &\mbox{if $D>d\ln (2d/\delta)$}\\
e\cdot \ln (2d/\delta) &\mbox{if $d\cdot\ln\delta^{-1}/e \le D\le d\ln (2d/\delta)$}\\
\frac{1.3\lnp{2d/\delta}}{\lnp{\frac{d}{D\cdot e}\lnp{2d/\delta}}}&\mbox{otherwise}
\end{cases} ,
$$ where $D$ is the number of distinct items in the input.
Consider storing fingerprints of size $f=\ceil{\log(d\cdot \mathcal M^2/\delta)}$
bits. Then with probability $1-\delta$ there are no false positives and the distinct operation terminates successfully.
\end{theorem}
\begin{proof}
Let $D_i$ denote the number of distinct items that are mapped into row $i$. Given two distinct items $x,y$, the probability that they share the same fingerprint is $2^{-f}$. Therefore, the probability of collisions between any two distinct items mapped into row $i$ is bounded by ${D_i \choose 2}2^{-f}\le D_i^2\cdot2^{-f-1}$.
Using the union bound, we have that the probability of any collision is at most $\sum_{i=1}^d D_i^2\cdot2^{-f-1}\le d\cdot2^{-f-1}\cdot \parentheses{\max\set{D_i}}^2\le \parentheses{\max\set{D_i}/\mathcal M}^2\cdot \delta/2$. Next, we can show that with probability $1-\delta/2$ the maximum load on each row is bounded as $\max\set{D_i}\le \mathcal M$. Therefore, we use the union bound once again to conclude that the probability of any within-row collision is at most $\delta$.
\begin{lemma}
$$
\Pr\brackets{\max\set{D_i} > \mathcal M}\le \delta/2.
$$
\end{lemma}
The proof of the light-load case (where $D<d\ln (2d/\delta)$) is similar to our analysis in Theorem~\ref{thm:topn} and we avoid repeating it for brevity. 
We use the following 
In the proof we use the following version of the Chernoff bound~\cite{Mitzenmacher:2005:PCR:1076315}.
\begin{lemma} (Chernoff Bound)
Let $X\sim\mathit{Bin}(n,p)$ be a binomial random variable with mean $np$, then for any $\gamma>0$:
    $$
        \Pr[X > np(1+\gamma)]\le \parentheses{\frac{e^\gamma}{(1+\gamma)^{1+\gamma}}}^{np}.
    $$
\end{lemma}
We set $\gamma=e-1, n=D, p=1/d$ to get that the probability that any of the rows see more than $\mathcal M$ distinct values is
$$
        d\cdot \parentheses{\frac{e^\gamma}{(1+\gamma)^{1+\gamma}}}^{D/d}=d\cdot e^{-D/d}\le \delta/2,
    $$
    where the last inequality follows from $D\ge d\ln (2d/\delta)$.
    We conclude a bound of $\mathcal M = D/d\cdot(1+\gamma) = e\cdot D/d$ as required.\qedhere
\end{proof}

We proceed with analysis of the pruning ratio on random order streams. Intuitively, if row $i$ sees $D_i$ distinct values and each are compared with $w$ that are stored in the switch memory, then with probability at least $w/D_i$ we will prune every duplicate entry. 
We use our maximum-load bound $\mathcal M$ for upper bounding the number of distinct items that are mapped to a given matrix row.
For example, consider a stream that contains $D=15000$ distinct entries and we have $d=1000$ rows and $w=24$ columns. Then we are expected to prune 58\% of the \emph{duplicated entries} (i.e., the entries that have appeared previously).
\begin{theorem}
Consider a random order stream with $D > d\ln(200d)$ distinct entries\footnote{It's possible to optimize other cases as well, but this seems to be the common case.}. Our algorithm, configured with $d$ rows and $w$ columns is expected to prune at least $ 0.99\cdot \min\set{\frac{w \cdot d}{D\cdot e},1}$ fraction of the duplicated entries.
\end{theorem}
\begin{proof}
We use the bound on the maximum number of distinct elements that are mapped into a row presented in Theorem~\ref{thm:distinct}. Specifically, we had that since $D > d\ln(200d)$ then with probability $99$\% we have that no row processes more than $e\cdot D/d$. That is, if we denote by $I$ this event, we have $\Pr[I]\ge0.99$. \emph{Conditioned on $I$}, we have that each redundant row is pruned with a probability of at least $\min\set{w/\max\set{D_i},1}\ge \min\set{\frac{w \cdot d}{D\cdot e},1}$. Therefore, we conclude that in expectation we prune at least $\Pr[I]\cdot \min\set{\frac{w \cdot d}{D\cdot e},1} = 0.99\cdot \min\set{\frac{w \cdot d}{D\cdot e},1}$.
\end{proof}


\section{Implementing Approximate Product in the Datapath}\label{app:skyline}
As discussed above, computing the product of the dimensions ($h_P(x)=\prod_{i=1}^D x_i$) in a \skyline query may be infeasible on the switch's datapath.
Instead, we propose to use the \emph{Approximate Product heuristic} (APH). 
Towards constructing APH, we first observe that for any $\beta>0$, $h_P(x)>h_P(y) \iff \sum_{i=1}^D\beta\log x_i > \sum_{i=1}^D\beta\log y_i$ by the monotonicity of the logarithm function. That is, $x$ dominates $y$ according to the product heuristic if and only if $\sum_{i=1}^D\beta\log x_i > \sum_{i=1}^D\beta\log y_i$. Unfortunately, the switch cannot compute logarithms either. Instead, we make a clever use of the switch's TCAM for computing the most significant $1$-bit in each dimension and approximate the logarithm in a fixed point representation with $\beta$ bits for the fractional part. The approximate logarithm values are then summed on the switch which proceeds as in the sum heuristic.



We utilize the match-action tables to approximate the log values. Specifically, we construct a static table with $2^{16}$ entries that maps each integer $a\in\set{1,2,\ldots,2^{16}}$ into its approximate logarithm form $f(a) = \brackets{\beta\log a}$. Here, $\beta>0$ is a parameter that determine an accuracy to representation-size tradeoff and $\brackets{\cdot}$ is the integer rounding operator. For example, if $\beta=2^{28}$, then the maximal image value is $\brackets{\beta\log (2^{16}-1)}<2^{32}-1$ and can thus be efficiently encoded using just 32-bits. 
APH is then defined as $\widehat{H_P(x)}\triangleq\sum_{i=1}^d f(x_i)$.
In case the data contains \inputCols{} wider than 16-bits (e.g., $32$ or $64$ bit integers) we can still use just $2^{16}$-sized match-action table by applying it on the 16-bits starting with the 1 in the representation. For example, for $z\in\set{0,\ldots,2^{32}-1}$ let $\ell$ denote the index of the most significant set bit in the binary representation of $z$ (i.e., $\ell=\floor{\log_2 z}$). 
If $\ell\le 16$ we can apply the table on the $16$ least significant bits of $z$. Otherwise, we apply it on $z_{\ell,\ldots,\ell-16}$. Specifically, if we denote by $z'$ these bits we have that $z\approx z'\cdot 2^{32-\ell}$ and thus $\log(z)\approx \log(z')+(32-\ell)$.
To find the value of $\ell$ we use the switch TCAM which using $32$ or $64$ rules can compute $\ell$ with a single \mbox{lookup for $32$ or $64$-bit integers respectively.}
\section{Analysis of the \topn Algorithm}\label{app:topn}
The goal of our algorithm is to ensure that with probability $1-\delta$, where $\delta$ is an error parameter set by the user, no more than $w$ \topn values are mapped into the same \matrixRow{}. In turn, this guarantees that the pruning operation is successful and that all \outputRows{} are not pruned. In the following, we assume that $d$ is given (this can be derived from the amount of per-stage memory available on the switch) and discuss how to set the number of \matrixCols{}.
To that end, we use $w\triangleq\floor{\frac{1.3\lnp{d/\delta}}{\lnp{\frac{d}{N\cdot e}\lnp{d/\delta}}}}$ \matrixCols{}. For example, if we wish to find the top-1000 with probability 99.99\% (and thus, $\delta=0.0001$) and have $d=600$ \matrixRows{} then we use $w=16$ \matrixCols{}. Having more space (larger $d$) reduces $w$; e.g., with $d=8000$ \matrixRows{} we require $5$ \matrixCols{}. Having too few \matrixRows{} may require an excessive number of \matrixCols{} (e.g., $w=288$ \matrixCols{} are required for $d=200$) which may be infeasible due to the limited number of pipeline stages. 
\begin{theorem}\label{thm:topn}
Let $d,N\in\mathbb N, \delta > 0$ such that $d\ge N\cdot e / \ln\delta^{-1}$ and define $w\triangleq\floor{\frac{1.3\lnp{d/\delta}}{\lnp{\frac{d}{N\cdot e}\lnp{d/\delta}}}}$. Then \topn query succeeds with probability at least $1-\delta$.
\end{theorem}
\begin{proof}
We use the following fact:
\begin{fact}\label{fact}
Let $x,y\in\mathbb R$ such that $y>e$ and $x\ge 1.3\ln y/\ln \ln y$ then $x^x\ge y$.
\end{fact}

Denote $x\triangleq\parentheses{\frac{(w+1)\cdot d}{N\cdot e}}$ and $y\triangleq  \parentheses{d/\delta}^{\frac{d}{N\cdot e}}$. 
%
The number of \topn elements that are mapped into \matrixRow{} $i$ is a binomial random variable $X_i\sim Bin(N,d)$. We wish to show that the probability that there exists a \matrixRow{} for which $X_i> w$ is at most $\delta$.
By using the union bound, we get that it is enough to show
$\Pr\brackets{X_i> w}\le \delta/d$.
We have that $$\Pr\brackets{X_i> w}\le{N \choose w+1}\cdot d^{-(w+1)}\le \parentheses{\frac{N\cdot e}{(w+1)\cdot d}}^{w+1}.$$
Our goal is to show that $\parentheses{\frac{N\cdot e}{(w+1)\cdot d}}^{w+1}\le \delta/d$, which is equivalent to showing $\parentheses{\frac{(w+1)\cdot d}{N\cdot e}}^{\frac{(w+1)\cdot d}{N\cdot e}}=x^x\ge \parentheses{\frac{d}{\delta}}^{\frac{d}{N\cdot e}}=y$.

Observe that $1.3\ln y/\ln \ln y = \frac{\frac{1.3d}{N\cdot e}\lnp{d/\delta}}{\lnp{\frac{d}{N\cdot e}\lnp{d/\delta}}}
= \frac{x}{w+1}\cdot \frac{1.3\lnp{d/\delta}}{\lnp{\frac{d}{N\cdot e}\lnp{d/\delta}}}
$ and thus $x\ge 1.3\ln y/\ln \ln y$ is equivalent to $w\ge\floor{\frac{1.3\lnp{ d/\delta}}{\lnp{\frac{d}{N\cdot e}\lnp{d/\delta}}}}$ which holds by the value we set to $w$. Finally, Fact~\ref{fact} shows that $x^x>y$ as required and concludes the proof.
\end{proof}
The above theorem shows how to configure the algorithm to ensure the correctness of the operation. 
However, correctness alone is achievable by not pruning any \inputRow{}. Therefore, it is important to asses the fraction of data that our algorithms prune.

In the worst case, if the switch has no feedback from the software stream processor, no pruning is possible. That is, if the input stream is monotonically increasing, the switch must pass all \inputRows{} to ensure correctness. 
In practice, streams are unlikely to be adversarial as the order in which they are stored is optimized for performance. 
To that end, we analyze the performance on random streams, or equivalently, arbitrary streams that arrive in a random order.
Going back to the above example, if we have $d=600$ \matrixRows{} on the switch and aim to find TOP $1000$ from a stream of $m=8M$ elements, our algorithm is expected to prune at least 99\% of the data. For a larger table of $m=100M$ entries our bound implies expected pruning of over 99.9\% 
Observe that the logarithmic dependency on $m$ in the following theorem implies that our algorithm work better for larger datasets. 
\begin{theorem}\label{sec:pruning}
Consider a random-order stream of $m$ elements and consider the \topn operation with algorithm parameters $d,w$ as discussed above. Then our algorithm prunes at least all but $w\cdot d\cdot{\lnp{\frac{m\cdot e}{w\cdot d}}}$ of the $m$ elements in expectation.
\end{theorem}
\begin{proof}
Denote by $M_i$ a random variable denoting the number of elements mapped into \matrixRow{} $i\in\set{1,2,\ldots,d}$. 
By the operation of the algorithm, the first $\min\set{M_i,w}$ elements are surely not pruned. Each following value is pruned unless it is one of the $w$-largest seen so far. Thus, the probability that the $j$'th element is not pruned, for $j\in\set{1,\ldots,M_i}$ is $\min\set{\frac{w}{j},1}$. Therefore, the \emph{expected} number of elements that are mapped to \matrixRow{} $i$ and are not pruned is $\sum_{j=1}^{M_i}\min\set{\frac{w}{j},1} = w +  \sum_{j=w+1}^{M_i}\frac{w}{j} = w\cdotpa{1+H_{M_i}-H_w} \le w\cdotpa{1+\lnp{M_i/w}}$, where $H_z$ is the $z$ harmonic number $H_z=\sum_{q=1}^z 1/q$.

Next, we take into account all $d$ \matrixRows{}. By linearity of expectation we have that the overall expected number of elements that are not pruned is bounded by $\mathbb E\brackets{\sum_{i=1}^d w\cdotpa{1+\lnp{M_i/w}}}$. Since $\sum_{i=1}^d M_i = m$, and using the concaveness of the logarithm function and Jensen's inequality, we conclude a bound on the number of non pruned elements of $w\cdot d\cdotpa{1+\lnp{\frac{m}{w\cdot d}}}=w\cdot d\cdot{\lnp{\frac{m\cdot e}{w\cdot d}}}$.
\end{proof}

\paragraph{Optimizing the Space and Pruning Rate}
The above analysis considers the number of \matrixRows{} $d$ as given and computes the optimal value for the number of \matrixCols{} $w$. 
However, unless one wishes to use the minimal number of \matrixCols{} possible for a given per-stage space constraint, we can simultaneously optimize the space and pruning rate.
To that end, observe that the required space for the algorithm is $\Theta(w\cdot d)$, while the pruning rate is monotonically increasing in $w\cdot d$ as shown in Theorem~\ref{sec:pruning}.
Therefore, by minimizing the product $w\cdot d$ we optimize the algorithm in both aspects.
Next, we note that for a fixed error probability $\delta$ the value for $w$ is monotonically decreasing in $d$ as shown in Theorem~\ref{thm:topn}. Therefore we define $f(d)\triangleq w\cdot d \approx { \frac{d\cdot1.3\lnp{d/\delta}}{\lnp{\frac{d}{N\cdot e}\lnp{d/\delta}}}}$ and minimize it over the possible values of $d$.\footnote{This omits the flooring of $w$ as otherwise the function is not continuous. The actual optimum, which needs to be integral, will be either the minimum $d$ for that value or for $w$ that is off by $1$.}
The solution for this optimization is setting $d\triangleq {\delta}\cdot e^{W(N\cdot e^2/\delta)}$, where $W(\cdot)$ is the Lambert $W$ function defined as the inverse of $g(z)=ze^z$. For example, for finding TOP 1000 with probability 99.99\% we should use $d=481$ \matrixRows{} and $w=19$ \matrixCols{}, even if the per-stage space allows larger~$d$.

\section{Extended Evaluation}

\subsection{Comparison with NetAccel}

We now compare {\NAME} with NetAccel~\cite{lerner2019case}, which also offloads  Join and GroupBy queries to switches. There are two key differences between \NAME and NetAccel: First, 
NetAccel stores query results on the switch and then sends it to the master when the query processing is complete. In contrast, \NAME prunes entries that are not part of the results and sends all other entries to the master in a streaming manner.

Second, when the switch does not support certain queries, NetAccel offloads the remaining processing to the switch CPU, while \NAME sends unpruned entries to the master.

Since NetAccel is a work-in-progress and does not have a full implementation for fair comparison, we run a microbenchmark to understand the impact of the two differences. 


\begin{figure}
    \centering
    \includegraphics[width=\linewidth]{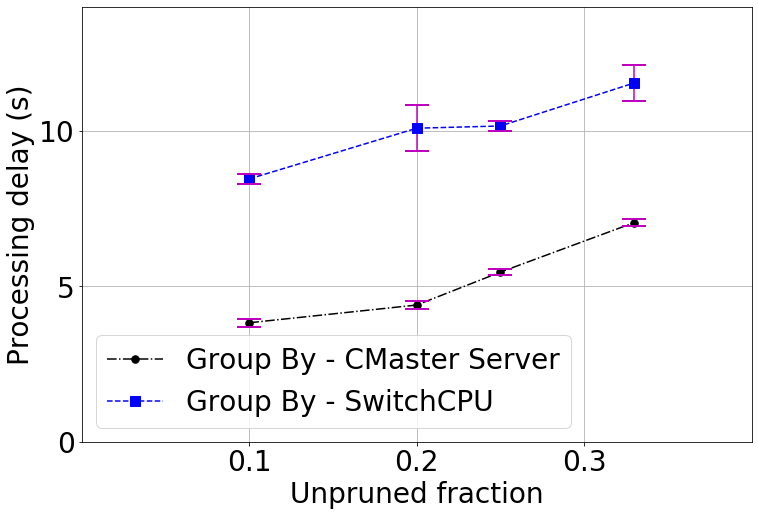}
    \caption{The time of processing the Group-By query on a server compared to processing it on the switch CPU.}
    \label{fig:group_by_switch_cpu}
\end{figure}

\begin{figure}
    \centering
    \includegraphics[width=\linewidth]{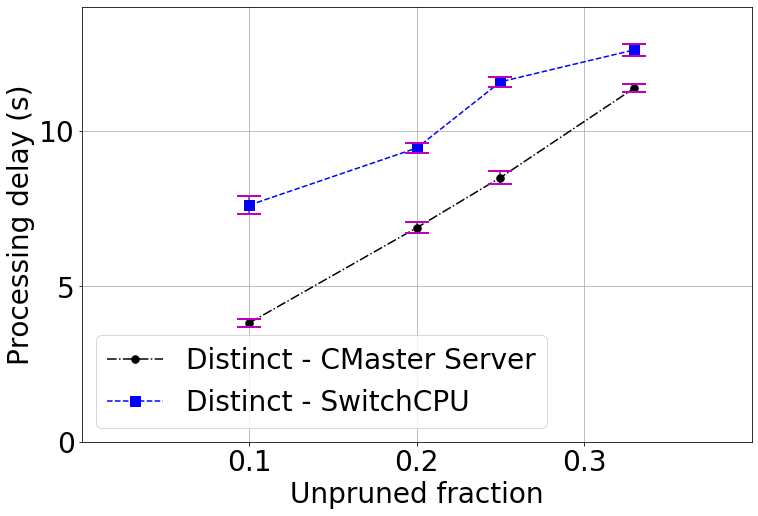}
    \caption{The time of processing the Distinct query on a server compared to processing it on the switch CPU.}
    \label{fig:distinct_switch_cpu}
\end{figure}

\begin{itemize}
    \item Pruning is a better design than storing result on the switch because:
    \begin{itemize}
        \item Storing the result in the switch instead of pruning makes it difficult to support a large number of queries hence NetAccel only supports \textsc{Join} and \textsc{Group-By}.
       \item Storing the result in the switch means NetAccel spends an extra pass at the end transferring the result to the master server. This increases the latency of the query because 1) extra time is spent in data movement of the result and 2) the next step in the query (e.g late materialization) can only be started at the end when it gets the result from the switch instead of pipelined.
    \end{itemize}
    \item Offloading remaining processing to the master server instead of offloading it to the switch CPU is a better idea because:
    \begin{itemize}
        \item The throughput between the data plane and the control plane on the switch is limited adding to overall query latency. 
        \item The switch CPU is not as powerful as a master server making it not scale well as the portion of data not processed in the dataplane increases (figure ~\ref{fig:group_by_switch_cpu} and figure ~\ref{fig:distinct_switch_cpu}).
    \end{itemize}
\end{itemize}

\fi

\end{document}